\documentclass{article}

\usepackage[english]{babel}

\usepackage[top=1in,bottom=1in,left=1in,right=1in]{geometry}

\usepackage[bookmarks=true,hypertexnames=false,pagebackref]{hyperref}
\hypersetup{colorlinks=true, citecolor=blue, linkcolor=red, urlcolor=blue}
\usepackage[textsize=tiny]{todonotes}

\usepackage{mathrsfs}
\usepackage{graphicx}

\usepackage{amsthm,amsmath,amssymb}
\usepackage{xifthen}
\usepackage[ruled,linesnumbered,vlined]{algorithm2e}

\usepackage{cleveref} 

\usepackage{float}
\usepackage{xspace}
\usepackage{subfigure}
\usepackage{bbm}
\usepackage{natbib}
\usepackage{thmtools} 
\usepackage{thm-restate}
\newtheorem{example}{Example}
\newtheorem{theorem}{Theorem}

\def\1#1{\mathbbm 1\left[#1\right]}
\def\~#1{\mathsf{#1}}
\def\*#1{\mathbf{#1}} \def\+#1{\mathcal{#1}} \def\-#1{\mathrm{#1}}\def\^#1{\mathbb{#1}}\def\!#1{\mathtt{#1}}\def\@#1{\mathscr{#1}}
\newcommand{\set}[1]{\left\{#1\right\}}

\renewcommand{\mid}{\;\middle\vert\;}

\renewcommand{\Pr}[2][]{ \ifthenelse{\isempty{#1}}
  {\mathbf{Pr}\left[#2\right]} {\mathbf{Pr}_{#1}\left[#2\right]} }
\newcommand{\E}[2][]{ \ifthenelse{\isempty{#1}}
  {\mathbf{E}\left[#2\right]}
  {\mathbf{E}_{#1}\left[#2\right]} }
\newcommand{\Var}[2][]{ \ifthenelse{\isempty{#1}}
  {\mathbf{Var}\left[#2\right]}
  {\mathbf{Var}_{#1}\left[#2\right]} }
\newcommand{\Ent}[2][]{ \ifthenelse{\isempty{#1}}
  {\mathbf{Ent}\left[#2\right]}
  {\mathbf{Ent}_{#1}\left[#2\right]} }

\renewcommand{\emptyset}{\varnothing}

\newtheorem{definition}[theorem]{Definition}
\newtheorem*{remark}{Remark}

\newtheorem{proposition}[theorem]{Proposition}

\newtheorem{lemma}[theorem]{Lemma}

\newcommand{\SIR}{\textnormal{SIR}\xspace}
\newcommand{\SIS}{\textnormal{SIS}\xspace}
\newcommand{\TSIR}{\textnormal{TSIR}\xspace}

\newcommand{\IC}{\textnormal{IC}\xspace}
\newcommand{\LT}{\textnormal{LT}\xspace}
\newcommand{\RR}{\textnormal{RR}\xspace}

\newcommand{\IMM}{\mathsf{IMM}\xspace}
\newcommand{\NODE}{\mathsf{NodeSelection}\xspace}
\newcommand{\InfMax}{\textnormal{InfMax}\xspace}
\newcommand{\CELF}{\mathsf{CELF}\xspace}
\newcommand{\DegDis}{\mathsf{DegDis}\xspace}
\newcommand{\VRP}{\mathsf{VoteRank}^{++}}

\newcommand{\SIRIMM}{\!{SIRIMM}\xspace}
\newcommand{\TSIRIMM}{\!{TSIRIMM}\xspace}

\DeclareMathOperator{\argmax}{argmax}

\newcommand*\samethanks[1][\value{footnote}]{\footnotemark[#1]}

\title{A Thorough Comparison Between Independent Cascade and Susceptible-Infected-Recovered Models\thanks{A short version of this paper appears in AAAI'25.}}

\author{Panfeng Liu\thanks{Shanghai Jiao Tong University \{liupf22, guoliang.qiu, bstao, kuan.yang\}@sjtu.edu.cn}
        \and Guoliang Qiu\samethanks[2]
	\and Biaoshuai Tao\samethanks[2]
	\and Kuan Yang\samethanks[2]}

\begin{document}
\date{}
\maketitle

\begin{abstract}
    We study cascades in social networks with the independent cascade ($\IC$) model and the Susceptible-Infected-recovered ($\SIR$) model.
    The well-studied IC model fails to capture the feature of \emph{node recovery}, and the SIR model is a variant of the IC model with the node recovery feature.
    In the SIR model, by computing the probability that a node successfully infects another before its recovery and viewing this probability as the corresponding IC parameter, the SIR model becomes an ``out-going-edge-correlated'' version of the IC model: the events of the infections along different out-going edges of a node become dependent in the SIR model, whereas these events are independent in the IC model.
    In this paper, we thoroughly compare the two models and examine the effect of this extra dependency in the SIR model.
    By a carefully designed coupling argument, we show that the seeds in the $\IC$ model have a stronger influence spread than their counterparts in the $\SIR$ model, and sometimes it can be significantly stronger.
    Specifically, we prove that, given the same network, the same seed sets, and the parameters of the two models being set based on the above-mentioned equivalence, the expected number of infected nodes at the end of the cascade for the IC model is weakly larger than that for the SIR model, and there are instances where this dominance is significant.

    We also study the influence maximization problem (the optimization problem of selecting a set of nodes as initial seeds in a social network to maximize their influence) with the SIR model.
    We show that the above-mentioned difference in the two models yields different seed-selection strategies, which motivates the design of influence maximization algorithms specifically for the SIR model.
    We design efficient approximation algorithms with theoretical guarantees by adapting the reverse-reachable-set-based algorithms, commonly used for the $\IC$ model, to the $\SIR$ model.
\end{abstract}


\section{Introduction}
\label{sec:Intro}
The study of information diffusion in social networks, such as Facebook, Twitter, and WeChat, has garnered significant attention in the fields of communication media and social science \citep{KKT15,DPRM2001,BJJR1987,RD2002,Rigobon2002,PV01,KR2010}.
Information diffusion is typically characterized by \emph{cascading}---a fundamental social network process in which a number of nodes, called \emph{seeds}, initially possess a certain attribute or piece of information and may spread to their neighbors.

Numerous diffusion models have been developed so far. Among them, the \emph{independent cascade model}~\citep{KKT15} is well-known and extensively studied.
In the independent cascade model (hereinafter denoted by $\IC$), each edge $(u,v)$ is assigned a probability $p_{u,v}$. An infected node $u$ attempts to infect its neighbor $v$ \emph{only once}, with a success probability of $p_{u,v}$. The events of successful infections along different edges are independent.
In the $\IC$ model, an infected node remains infected throughout the cascade.

However, in many real-world scenarios, an infected node may \emph{recover}.
This phenomenon can be observed in various situations. For example, consider when an individual subscribes to a magazine or signs up for a fitness membership card under the influence of friends; eventually, this subscription or membership may be terminated. This termination could result from factors such as the user losing interest in the service or the subscription/membership expiring. Another example arises when someone introduces a new product to a friend; initially, they may promote it, but eventually, they tire of advertising it after a few days.
Similarly, when a user shares a post on Twitter, their friends may retweet it only for a short period before it becomes overshadowed by other posts.
The recovery of nodes may signify the loss of a certain attribute or the end of spreading this attribute, both of which undoubtedly impact the cascading process.

The concept of node recovery is effectively captured by \emph{epidemic models}. 
These models divide individuals into distinct states, including \emph{susceptible}, \emph{infected}, and \emph{recovered}. Various models have been formulated based on feasible state transitions \citep{kermack1927contribution,ALLEN1994,Priscilla2009}. 
The \emph{Susceptible-Infected-Recovered model}~\citep{kermack1927contribution} (hereinafter denoted by $\SIR$) includes the process from susceptible to infected and eventually to recovery and permanent immunity. 
In the $\SIR$ model, each node $u$ has a \emph{recover rate} $\gamma_u$. 
In each round, an infected node $u$ attempts to infect each neighbor $v$, succeeding with probability $\beta_{u,v}$, and then the node $u$ recovers with probability $\gamma_u$.
If $u$ does not recover, it remains infected and continues to attempt to infect its neighbors in the next round and subsequent rounds until it recovers.
Once $u$ recovers, it can no longer be infected.
Obviously, after a sufficiently long time, all infected nodes will be recovered, and nodes can only be either susceptible or recovered.
From the applicational perspective, we are interested in the number of the nodes \emph{that have been infected} (e.g., in the examples in the previous paragraph, the advertiser immediately receives benefits when the users pay the subscription fees/membership fees).
Equivalently, we are interested in the number of the \emph{recovered} at the end of the cascade.

\paragraph{Plausible similarity between IC and SIR models.}
Previous work has observed the similarity between $\IC$ and $\SIR$ (see, e.g., Chapter 8.2 in \cite{chen2020network}).
In the $\SIR$ model, the probability that node $u$ infects a neighbor $v$ in some round is given by
\begin{equation}\label{eqn:IC-SIR}
p_{u,v}=\sum_{t=1}^\infty \gamma_{u}(1-\gamma_u)^{t-1}\left(1-(1-\beta_{u,v})^t\right).
\end{equation}
Here, $\gamma_u(1-\gamma_u)^{t-1}$ represents the probability of node $u$ recovering in round $t$, and $(1-(1-\beta_{u,v})^t)$ is the probability of node $u$ successfully infecting $v$ in at least one of the $t$ rounds.
By considering $p_{u,v}$ in \Cref{eqn:IC-SIR} as the $\IC$ parameter, we can observe the similarity in the cascading processes between both models. In the $\SIR$ model, we can ``aggregate'' the infection attempts along each edge over multiple rounds, and the overall probability becomes the $\IC$ parameter.

However, this observation raises a question and thus leads to the distinction between the $\IC$ model and the $\SIR$ model. In the $\SIR$ model, due to this aggregation, infections along the incident edges to a node are no longer independent.
Intuitively, the event of $u$ successfully infecting one of its neighbors $v_1$ is positively correlated with the event of $u$ infecting another neighbor $v_2$,
as success on one edge increases the likelihood of $u$ not recovering for a longer period, thus enhancing the probability of success on another edge. 
This dependency creates a difference between the two models. Our objective is to explore whether this additional dependency on outgoing edges enhances or hampers the spread of information. Consequently, we consider the $\SIR$ model as an outgoing-edges-dependent variant of the $\IC$ rather than following the original $\SIR$ definition.

From the above analysis, we can interpret the $\SIR$ model by the $\IC$ model if the above dependency can be ignored.
Past literature has noticed the similarity and this seemingly ``slight'' difference between the two models.
However, it is unclear how significant this difference is.
For example, given the same graph and the same seed set, and assuming that the parameters setting in $\IC$ and $\SIR$ satisfy \Cref{eqn:IC-SIR}, is the spread of the seeds in these two models substantially different?

\paragraph{Influence maximization with SIR model.}
The \emph{influence maximization problem} ($\InfMax$), proposed by \cite{KKT15,RD2002,DPRM2001}, is the problem of selecting a set of ``influencers'' (also known as \emph{seeds}) in social networks to maximize the expected number of influenced agents.
The $\InfMax$ problem has attracted remarkable attention from researchers due to its wide range of applications, including 
social advertising \citep{CJ2011,DPRM2001,RD2002}, product adoption \citep{BJJR1987,Bass1976,MMF1990,TKE2011}, disease analysis \citep{ZLSZ2018}, rumor control \citep{WP2017} and social computing \citep{LLWF22}.
More comprehensive literature reviews can be found in books~\citep{chen2020network,chen2022information} and survey~\citep{li2018influence}.

The $\InfMax$ problem with the $\IC$ model is well-studied.
In the first paper where the model is proposed, \cite{KKT15} show that a simple greedy algorithm achieves a $(1-1/e)$-approximation, and no polynomial time $(1-1/e+\epsilon)$-approximation algorithm exists assuming $\text{P}\neq\text{NP}$.
After this, extensive work has focused on designing faster algorithms while keeping the theoretical approximation guarantee~\citep{BBCL2014, CPL2012, GLL2011, TSX2015, LKGFVG2007, KKT15}.
One of the most successful types of such algorithms is based on \emph{reverse-reachable sets} \citep{BBCL2014,TYXY14,TSX2015,CW18}, which yields a near-linear time algorithm with the $(1-1/e)$ approximation ratio.

However, for those epidemic models including the $\SIR$ model, most of the existing results have focused on the \emph{dynamic} of these epidemic models and analyzed equilibrium states \citep{NS20, KKT19,EGK19,LW20,KAKKSBS20}. However, few studies have explored them \emph{in the context of influence maximization}.
On the other hand, maximizing the number of nodes that \emph{have been infected at least once} remains a natural problem in multiple real-world applications.
For example, in viral marketing, advertisers receive benefits or payments based on nodes' infection, irrespective of subsequent recoveries.
Nonetheless, nodes' recoveries can negatively affect the cascading process. Therefore, advertisers must consider this factor to attain a more comprehensive understanding of the cascade, enabling them to identify the optimal initial influencers.

We have previously seen the similarity and the subtle difference between the $\IC$ model and the $\SIR$ model.
Regarding $\InfMax$, are the seeding strategies considerably different for the two models even if the parameters satisfy \Cref{eqn:IC-SIR}?
If the difference is insignificant, the $\InfMax$ algorithms for the $\IC$ model can be directed applied to the $\SIR$ model by reducing the $\SIR$ model to the $\IC$ model based on \Cref{eqn:IC-SIR}.
Otherwise, we should look for algorithms that are specifically designed for the $\SIR$ model.

\subsection{Our Contributions}

\paragraph{Stronger propagation effect in the IC model.} 
As mentioned earlier, in the SIR model, through \Cref{eqn:IC-SIR}, we can compute the probability that a vertex $u$ ``eventually'' infects $v$ before its recovery and obtain an ``aggregated probability'' $p_{u,v}$.
This establishes an equivalence between the $\SIR$ model and the $\IC$ model, except that the events that the infections succeed along different edges become dependent in the $\SIR$ model.
We intensively study how this seemingly minor difference in dependency affects the influence spread and the seeding strategy.

Firstly, in \Cref{sect:liveandRRset,section:IC-and-SIR}, we show that this dependency can only \emph{harm} the influence spread.
Specifically, we prove the following novel observation: when comparing the two models, given the same graph, the same seed set, and the cascade model parameters related by \Cref{eqn:IC-SIR}, the expected number of infected vertices under the $\IC$ model is always weakly larger than that under the $\SIR$ model.
This is proved by developing a novel coupling method.

Furthermore, we show that the above gap in the expected number of infections can be made arbitrarily large in some instances.
In addition, we also observed that, in certain networks, the optimal seeding strategies for the $\SIR$ and $\IC$ models are different, leading to significant differences in propagation outcomes.
This motivates the need for designing algorithms specifically for the $\SIR$ model.

\paragraph{Approximation algorithms for the $\SIR$ model and its variants.}
We remark that the influence spread functions for all the proposed epidemic models based on $\SIR$ are \emph{submodular} functions.
This implies that the simple greedy algorithm achieves a $(1-1/e)$-approximation. 
However, the conventional greedy algorithm based on Monte-Carlo sampling is known to be slow in practice. We adapt those reverse-reachable-set-based algorithms to the $\SIR$ setting, which yields a $(1-1/e-\varepsilon)$-approximation algorithm with probability at least $1-n^{-\ell}$ for graphs with $n$ nodes and parameters $\varepsilon,\ell$, where the algorithm's running time is \emph{near-linear} in the number of nodes $n$.



\subsection{Related Work}
Many of the related papers have been discussed previously.
In this section, we discuss further related work on the algorithmic aspect.

In the most general case, approximating $\InfMax$ to within a factor of $n^{1-\epsilon}$ is NP-hard~\citep{KKT15}. 
The strong inapproximability comes from the \emph{nonsubmodularity} of the diffusion models: strong inapproximability results are known for $\InfMax$ even under very specific simple nonsubmodular diffusion models; as a result, researchers focus on heuristics (see, e.g., \citep{KKT15,chen2009approximability,angell2017don,schoenebeck2019beyond,schoenebeck2022think}).
On the other hand, when dealing with \emph{submodular} diffusion models, a simple greedy algorithm achieves an $(1 - 1/e)$-approximation:
specifically, the global influence function $\sigma(\cdot)$ is a submodular set function if the diffusion model is submodular~\citep{MR2010}, so the greedy algorithm achieves an $(1-1/e)$-approximation~\citep{nemhauser1978analysis}.
The approximation guarantee can be even slightly better than $(1-1/e)$ if the network is undirected~\citep{khanna2014influence,schoenebeck2020limitations}.
On the inapproximability side, it is known that the $(1-1/e)$-approximation cannot be improved for the $\IC$ model and directed graphs, unless $\text{P}=\text{NP}$~\citep{KKT15}.
Weaker APX-hardness results are known for undirected graphs and other submodular models such as the \emph{linear threshold model}~\citep{schoenebeck2020influence}.

Unfortunately, the standard greedy algorithm is limited by its scalability, and therefore a large number of $\InfMax$ algorithms have been designed and tested on real networks. As a result, numerous solutions have been developed to optimize the efficiency of reaching the broadest audience possible on a massive network.

Various heuristic metrics have been developed based on the topological attributes of nodes. For example, closeness centrality measures the importance of each node by computing the reciprocal of its average distance from other nodes \citep{SG1966}. Degree centrality, on the other hand, quantifies the influence of a node based on the number of its direct neighbors \citep{Freeman1978}. For instance, the $\DegDis$~\citep{WYS09} measures the influence of each node by iteratively selecting the nodes with the highest value that based on degrees, then discount the seeds' neighbors to obtain $k$ seeds. A similar idea was applied in the $\VRP$ algorithm \citep{LLFY21}, which utilizes a voting mechanism to measure the influence of nodes in networks. After selecting the node with the highest \emph{score}, the $\VRP$ discounts the \emph{voting ability} of the seed node's neighbors to reduce the overlap of influence between seed nodes. This enables the influence of the seed nodes to be spread as widely as possible.

However, most of the existing techniques either sacrifice approximation guarantees for practical efficiency, or the other way around. By leveraging submodularity, the $\CELF$ algorithm has been developed to efficiently handle large-scale problems. It achieves near-optimal placements while being 700 times faster than a simple greedy algorithm \citep{LKGFVG2007}. The methods in references ~\citep{BBCL2014,CPL2012,GLL2011,TSX2015,KKT15} that provide $\left(1-1/e-\epsilon \right)$-approximation algorithms with a cost that is lower than greedy method. 

One notable approach among these methods is the $\IMM$ algorithm (Influence Maximization with Martingales), introduced by \cite{TSX2015}. This algorithm operates with an expected time complexity of \(O\left((k+l) \cdot (n+m) \cdot \log n / \varepsilon^2\right)\) and provides an \(\left(1-1/e-\varepsilon \right)\)-approximation with a probability of at least \(1 - 1 / n ^ \ell\) under the triggering model. It addresses $\InfMax$ challenges by employing a martingale method, specifically estimating the number of reverse reachable sets. This feature allows $\IMM$ to significantly reduce redundant computations that were deemed unavoidable in the reference~\citep{TYXY14}.


Other perspectives on studying the $\InfMax$ problem, such as learning-based, adaptivity, data-driven, and reinforcement-learning-based approaches, are discussed in references~\citep{WLWCW19,peng2019adaptive,chen2019adaptivity,d2020improved,LKTLC21,chen2022adaptive,KXWYL23,MSLLLLN22,li2022link,d2023better}.

\subsection{Structure of This Paper}
In~\Cref{preliminaries}, we define the two diffusion models studied and the influence maximization problem.
In~\Cref{sect:liveandRRset}, we discuss the live-edge graph formulation of the $\IC$ and $\SIR$ models, and we adapt the reverse-reachable set technique to the $\SIR$ model.
These will be used in all the later sections.
Our results for the theoretical comparison of the $\IC$ and $\SIR$ models are in~\Cref{section:IC-and-SIR}.
In~\Cref{sec:SIRAlgorithms}, we discuss the algorithm design from the aspect of the $\InfMax$ problem with the $\SIR$ model. 
\section{Model and Preliminaries \label{preliminaries}}


    A social network can be represented as a directed graph $G=(V,E)$ with $n$ nodes (i.e., \emph{users}) and $m$ directed edges (i.e., \emph{social connections} between \emph{users}). For any directed edge $\left( u, v \right) \in E$, we say $\left( u, v \right) $ is an incoming edge (resp. outgoing edge) of $v$ (resp. $u$). We also call $u$ an incoming neighbor of $v$ and $v$ an outgoing neighbor of $u$.

    \begin{definition}[Diffusion Model and Influence Spread \citep{KKT15}]\label{def:diffusion}
       Given a social network (directed graph) $G=(V,E)$, a \textbf{diffusion model} $\Gamma$ is a (possibly random) function that maps from a vertex set $S$ (the seeds that are initially infected) to a vertex set $\Gamma_G(S)$ (the set of influenced vertices at the end of the spreading). 
       We omit the subscript $G$ when there is no confusion.
       Denote by $\sigma(S)=\E{|\Gamma(S)|}$ the expected number of vertices influenced by $S$.
    \end{definition}
    
    The goal of the \emph{influence maximization} problem ($\InfMax$) is to select at most $k$ nodes as \emph{seeds} to maximize the number of \emph{influenced} nodes on the social network $G$.
    
    \begin{definition}[Influence Maximization~\citep{KKT15}]
         Given a social network (directed graph) $G=(V,E)$, a diffusion model $\Gamma$, and a positive integer $k$, the objective of influence maximization is to select a subset $S\subseteq V$ with $|S|\leq k$ that maximizes the expected influence spread $\sigma(S)$.
    \end{definition}

    
\subsection{Independent Cascade Model}

In the Independent Cascade (IC) model~\citep{KKT15}, the nodes could be \emph{active} or \emph{inactive} in a given directed graph $G=(V,E)$. A node $v\in V$ could be activated by each of its incoming active neighbors independently. More precisely, each directed edge $e=\left(u, v\right)\in E$ is associated with an activation probability $p_e=p_{u,v}\in [0,1]$. The influence spread of an active seed set $S\subseteq V$ unfolds in discrete timestamps as follows.
\begin{enumerate}
    \item At timestamp $0$, only nodes in $S$ are active.
    \item At each timestamp $t=1,2,\dots$, each newly activated node $u$ from the previous timestamp gets one chance to activate its inactive outgoing neighbors; and for each inactive outgoing neighbor $v$, $u$ tries to activate $v$ with a probability $p_{u,v}$. The attempts to activate neighbors are independent of each other. If multiple incoming neighbors of an inactive node attempt to activate it, each attempt is considered separately with its own probability.
    \item The diffusion process terminates when no inactive node gets activated in a timestamp. 
\end{enumerate}
In the above process, once a node becomes active, it remains active throughout.
When we are considering the IC model, i.e., $\Gamma=\text{IC}_{\boldsymbol{p}}$ where $\boldsymbol{p}=\set{p_e}_{e\in E}$ in~\Cref{def:diffusion}, $\IC_{\boldsymbol{p},G}(S)$ is the set of active nodes at the end of the above diffusion process.



In~\cite{KKT15}, it is shown that the influence spread of the $\IC$ model can be simulated by evaluating the number of \emph{reachable} nodes from the seed set in the \emph{live-edge graph}. To be specific, the random live-edge graph $\+G_{\IC}(G,\boldsymbol{p})$ corresponding to the instance $(G,\boldsymbol{p})$ under the $\IC$ model is generated by including each directed edge $e\in E$ in $G$ with probability $p_e$. Each node $v$ is active if it is reachable from the seed set $S$, that is, there exists a directed path from some nodes $v\in S$ to $u$ in the live-edge graph.
More discussions about live-edge graphs can be found in~\Cref{sect:liveandRRset}.

\subsection{Susceptible-Infected-Recovered Model}



 \noindent In the Susceptible-Infected-Recovered Model ($\SIR$) model~\citep{kermack1927contribution}, the nodes could be \emph{susceptible}, \emph{infected}, or \emph{recovered}. 
 The $\SIR$ diffusion process is characterized by a directed graph $G=(V,E)$ together with two sets of parameters $\boldsymbol{\beta}=\set{\beta_e}_{e\in E}$ and $\boldsymbol{\gamma}=\set{\gamma_v}_{v\in V}$, where each node $v\in V$ is assigned a recovery probability $\gamma_v \in (0,1]$ and each directed edge $e=(u,v)$ is associated with an infection probability $\beta_e=\beta_{u,v} \in (0,1]$. The influence spread of an infected seed set $S\subseteq V$ unfolds in discrete timestamps as follows:
 \begin{enumerate}
     \item At timestamp $0$, all nodes in the seed set $S$ are initially infected, while the remaining nodes are considered susceptible.
     \item At each timestamp $t=1,2,\dots$, each node $u$ that is infected at timestamp $t$ performs the following operations sequentially:
     \begin{itemize}
         \item for each susceptible outgoing neighbors $v$, $u$ infects $v$ with probability $\beta_{u,v}$;
         \item $u$ gets recovered with a recovery probability $\gamma_u$ and remains infected otherwise.
     \end{itemize}
     \item The diffusion process terminates when all nodes are either recovered or susceptible. 
 \end{enumerate}
When we are considering the $\SIR$ model, i.e., $\Gamma=\SIR_{\boldsymbol{\beta},\boldsymbol{\gamma}}$ in~\Cref{def:diffusion}, $\SIR_{\boldsymbol{\beta},\boldsymbol{\gamma},G}(S)$ is the set of recovered nodes at the end of the above diffusion process.

 \begin{remark}
    It is implied that once a node becomes recovered in the $\SIR$ model, it can never be infected again or infect other nodes.
    Note that all infected nodes will eventually be recovered after a sufficiently long time.
 \end{remark}
 


We also consider the spread of influence within a specific time frame $T$ in the $\SIR$ model, referred to as the \emph{Truncated Susceptible-Infected-Recovered} $(\TSIR)$ Model. 
In this model, with a seed set $S\subseteq V$ infected at time $0$, the diffusion process unfolds as described earlier but stops at time $T$. 
The influence spread $\TSIR_{\boldsymbol{\beta},\boldsymbol{\gamma},T,G}(S)$ is defined as the number of nodes that have been infected by time $T$.
Thus, when taking $\Gamma=\TSIR_{\boldsymbol{\beta},\boldsymbol{\gamma},T}$, the ``influenced vertices'' in \Cref{def:diffusion} are those that are infected or recovered.

\section{Live-edge Graph and Reverse Reachable Set\label{sect:liveandRRset}}
Similar to the equivalent formulation proposed by \citet{KKT15} for the $\IC$ model, we demonstrate that the influence spread of the diffusion models we consider can be formulated using a model-specific live-edge graph. This live-edge graph formulation allows us to study the models using the unified characterization of the \emph{reverse reachable set}~\citep{BBCL2014} in our later discussion.

\subsection{Live-edge Graph Formulation}

The live-edge graph of a diffusion model is a random spanning sub-graph of $G$ that reflects the diffusion behavior. Specifically, each edge in the original network is either ``\emph{live}'' (active) or ``\emph{blocked}'' (inactive), based on the model's spread probabilities, and the live-edge graph is the subgraph consisting of live edges. 
Roughly speaking, each edge $(u, v)$ is live (included in the live-edge graph) with a probability, which represents the likelihood of $u$ activating $v$ during the influence spread period.
Thus, a live-edge graph is an equivalent representation of a spreading process on $G$ based on a diffusion model, and it serves as a sampling graph that captures one possible propagation of the model. By evaluating the number of \emph{reachable} nodes from a given seed set $S$ in the \emph{live-edge graph}, we can simulate the spreading ability of $S$ in the diffusion process.

\begin{definition}[Live-edge Graph Formulation] Given a social network (directed graph) $G=(V,E)$ with the diffusion model $\Gamma$, we say $\Gamma$ is a live-edge graph diffusion model if there exists a measure $\mu$ over the collection of (possibly edge-weighted) spanning sub-graphs of $G$, such that the influence spread $\sigma(S)$ equals the expected number of reachable nodes from $S$ in the live-edge graph $\+G$ sampled from $\mu$ for any seed set $S\subseteq V$.
\end{definition}

For convenience, we abuse $\+G$ to represent the edge set $E(\+G)$ of the live-edge graph $\+G$ and thus $e\in \+G$ means that $e$ is one of the edges in $\+G$ for any $e\in E$ in subsequent discussion.


In the following, we introduce the live-edge graph formulation for the $\SIR$ model.
For the $\TSIR$ model, the live-edge graph characterization can be found in \Cref{sec:LiveedgeTSIR}.

%


\subsubsection{The Live-edge Graph for the $\SIR$ Model}\label{subsec:live-edge-SIR}

The influence spread of seed nodes in the $\SIR$ model on instance $\left(G=(V,E),\boldsymbol{\beta}, \boldsymbol{\gamma}\right)$ can be characterized by the number of reachable nodes in the live-edge graph
\[\+G_{\SIR}=\+G_{\SIR}(G,\boldsymbol{\beta},\boldsymbol{\gamma})=\+G_{\SIR}\left(G,\set{\boldsymbol{R}_v}_{v\in V},\set{\boldsymbol{I}_e}_{e\in E}\right)\] induced by a collection of independent variables
$$\boldsymbol{R}_v=\left(R_{v,1},R_{v,2},\dots\right),\quad\mbox{and}\quad \boldsymbol{I}_e=\left(I_{e,1},I_{e,2},\dots\right),$$ 
where $R_{v,t}\sim \!{Bern}(\gamma_v)$ for any $v\in V, t\in \^N^+$ and $I_{e,t}\sim \!{Bern}(\beta_e)$ for any $e\in E, t\in \^N^+$ are Bernoulli random variables. 
Here, $R_{u,t} =1$ denotes that $u$ is recovered at the $t$-round after $u$'s infection, and $R_{v,t}=0$ otherwise. Similarly, $I_{(u,v),t}$ is the indicator random variable for the event that $u$ successfully infects $v$ at the $t$-th round after $u$ is infected.
Specifically, each directed edge $ e= \left( u, v \right) \in E$ is included in $\+G_{\SIR}$ if there exists some $t^*\geq 1$ such that $I_{e,t^*}=1$ and $R_{u,[t^*-1]}=( R_{u, 1}, R_{u, 2},...,R_{u, t^*-1})$ is a sequence of zeros with length $t^*-1$.
In other words, the influence spread from node $u$ to node $v$ succeeds at timestamp $t^*$ before the node $u$ gets recovered, and the edge $e$ represents a successful infection event between nodes $u$ and $v$ that could have occurred. Here, we say a node $v$ is reachable from $u$ in $\+G_{\SIR}$ if there exists a directed path from $u$ to $v$ in $\+G_{\SIR}$.
An example is illustrated below.

\begin{example}
\emph{
Consider a graph $G=(V,E)$ with $V=\{u,v,w\}$ and $E=\{(u,v),(u,w)\}$. Assume we have sampled indicator random variables $\boldsymbol{R}_{u}=(R_{u,1},R_{u,2},\ldots) = (0,0,0,1,\ldots)$, $\boldsymbol{I}_{(u,v)}= (I_{(u,v),1},I_{(u,v),2},\ldots) = (0,0,1,\ldots)$, and $ \boldsymbol{I}_{(u,w)}= (I_{(u,w),1}, I_{(u,w),2},\ldots) =(0,0,0,0,0,1,\ldots)$.  In this example, $u$ is recovered at round $4$. The infection from $u$ to $v$ is successful at round $3$, and the infection from $u$ to $w$ is successful at round $6$.
Therefore, $u$ can successfully infect $v$ before $u$'s recovery, while $u$ cannot infect $w$ before its recovery.
Correspondingly, in the live-edge graph, $(u,v)$ is live and $(u,w)$ is blocked, i.e., $(u,v) \in \mathcal{G}_{\text{SIR}}$ and $(u,w) \notin \mathcal{G}_{\text{SIR}}$.}
\end{example}

The proof of the following proposition is straightforward: this is just a rephrasing of the same stochastic process. It is deferred to \Cref{sec:appendix}. 

\begin{restatable}{proposition}{liveedgeSIR}\label{prop:live-edge-graph-SIR}
    $\+G_\SIR$ defined above is a live-edge graph formulation of $\textnormal{SIR}_{\boldsymbol{\beta},\boldsymbol{\gamma}}$, namely, 
    $$\sigma_{\SIR}(S)=\E{\mbox{the number of reachable nodes from $S$ in $\+G_{\SIR}$}}.$$
\end{restatable}

As mentioned in \Cref{sec:Intro}, we are interested in comparing the two diffusion processes corresponding to $\IC$ and $\SIR$ respectively with the equal marginal probability for $u$ successfully infecting $v$ along each edge $(u,v)$.
\Cref{eqn:IC-SIR} ensures this, which is described in the proposition below (whose proof is straightforward).

\begin{proposition}
    Given any directed graph $G=(V,E)$ together with the diffusion models $\IC_{\boldsymbol{p}}$ and $\SIR_{\boldsymbol{\beta,\gamma}}$ where the parameters satisfying \Cref{eqn:IC-SIR}, it holds that $\Pr{e\in \+G_{\IC}}=\Pr{e\in \+G_{\SIR}}$ for each $e\in E$.
\end{proposition}

\subsubsection{The Live-edge Graph for the $\TSIR$ Model}\label{sec:LiveedgeTSIR}
Given a specific time threshold $T$, the influence spread of seed nodes in the $\TSIR$ diffusion model on instance $(G=(V,E),\boldsymbol{\beta},\boldsymbol{\gamma},T)$ can be characterized by the number of reachable nodes in the live-edge graph 
\begin{align*}
    \+G_{\TSIR}\!=\!\+G_{\TSIR}(G,\boldsymbol{\beta},\boldsymbol{\gamma}, T)=\+G_{\TSIR}\bigl(G,\!\{\boldsymbol{R}_v\}_{v\in V}\!,\!\{\boldsymbol{I}_e\}_{e\in E}\!,T\bigr)
\end{align*}
induced by a collection of independent Bernoulli random variables, similar to that of the $\SIR$ model together with the weight on the edges: an edge $e= \left( u, v \right) \in E$ is included in $\+G_{\TSIR}$ with infection span $w_e=w_{u,v}=t^*$ if $t^*$ is the smallest positive integer satisfying $I_{e,t^*}=1$ and $R_{u,[t^*-1]}$ is a sequence of zeros with length $t^*-1$. 

The infection span of each edge $(u,v)$ captures the timestamps needed to influence propagation from $u$ to $v$.
We call the node $v$ is $T$-reachable from $u$ in $\+G_{\TSIR}$ if there exists a directed path $P=(u_1=u,u_2,\dots,u_\ell=v)$ in $\+G_{\TSIR}$ such that $\sum_{i\in \ell-1} w_{u_i,u_{i+1}}\leq T$. 
The influenced nodes of any seed set $S\subseteq V$ on the graph $G=(V,E)$ then equals the number of $T$-reachable nodes from $S$ in $\+G_{\TSIR}$.  
An example is illustrated below.

\begin{example}
\emph{
Consider a graph $G=(V,E)$ with $V=\{u,v,w\}$ and $E=\{(u,v),(u,w),(v,w)\}$. Assume we have sampled indicator random variables $\boldsymbol{R}_{u}=(R_{u,1},R_{u,2},\ldots) = (0,0,0,1,\ldots)$, $\boldsymbol{I}_{(u,v)}= (I_{(u,v),1},I_{(u,v),2},\ldots) = (0,0,1,\ldots)$, $\boldsymbol{I}_{(u,w)}= (I_{(u,w),1}, I_{(u,w),2},\ldots) =(0,0,0,0,0,1,\ldots)$, $\boldsymbol{I}_{(v,w)}= (I_{(u,v),1},I_{(u,v),2},\ldots) = (0,1,\ldots)$, and we set $T=4$.
In the above example, $u$ is recovered at round $4$, $u$ infect $v$ at round $3$, $u$ infect $w$ at round $6$, and $v$ infect $w$ at round $2$. In the graph $\+G_{\TSIR}$, the edges $(u,v)$ and $(v,w)$ have infection spans of $3$ and $2$ respectively. Consequently, within the time threshold $T=4$, $u$ can successfully infect $v$, but it will fail to infect $w$ due to the shortest path from $u$ to $w$ having a total weight of $5$.}

\end{example}


Similarly, we can prove the following live-edge graph characterization for the $\TSIR$ model, the proof of which is similar to that of~\cref{prop:live-edge-graph-SIR}.
\begin{proposition}
    Given a social network (directed graph) $G=(V,E)$ with the $\TSIR_{\boldsymbol{\beta},\boldsymbol{\gamma},T}$ diffusion model, we have for any seed set $S\subseteq V$,
    \[\sigma_{\TSIR}(S)=\E{\mbox{the number of $T$-reachable nodes from $S$ in $\+G_{\TSIR}$}}.\]
\end{proposition}

\subsection{Reverse Reachable Set Characterization}
The influence spread of the aforementioned models can be characterized by the notion of the \emph{reverse reachable set} \citep{BBCL2014} under the live-edge graph formulation. 
 When considering a live-edge graph diffusion model, we sometimes use $\+G$ to represent its diffusion process instead of $\Gamma$. 
The definition of the reverse reachable set can be stated as follows.
\begin{definition}[Reverse Reachable Set]
Given a directed graph $G=(V,E)$ with the live-edge graph $\+G$, the reverse reachable set of a node $v \in V$, denoted by $\RR_{\+G}(v)$, is the set of all nodes in $\+G$ that can reach $v$. Furthermore, let $\RR_{\+G}$ be the random set $\RR_{\+G}(v)$ with $v$ selected uniformly at random from $V$.
\end{definition}


\begin{proposition}\label{lem:reverse-reachable-set-characterization}
    Given a directed graph $G=(V,E)$ with a live-edge graph diffusion model $\+G$, the probability that the diffusion process from any seed set $S\subseteq V$ can influence a node $v$ equals the probability that $S$ overlaps with the set $\RR_{\+G}(v)$, i.e., $\Pr{S\cap \RR_{\+G}(v)\neq \emptyset}$. Furthermore, the influence spread satisfies \[\sigma(S)=\sum\nolimits_{v\in V}\Pr{{S\cap \RR_{\+G}(v)}\neq \emptyset}=|V|\cdot\Pr{{S\cap \RR_{\+G}}\neq \emptyset} .\]
\end{proposition}

In later discussion, we sometimes use $\RR_{\IC}(v)$ ($\RR_{\IC}$, resp.) to denote the reverse reachable set of a node $v$ (random reverse reachable set, resp.) corresponding to the live-edge graph $\+G_{\IC}$ for simplicity. The notations $\RR_{\SIR}(v)$, $\RR_{\TSIR}(v)$, $\RR_{\SIR}$, and $\RR_{\TSIR}$ are defined similarly.


\section{Theoretical Comparison between the $\IC$ and $\SIR$ Models}
\label{section:IC-and-SIR}

As discussed in~\Cref{sec:Intro}, we can ``aggregate'' the infection attempts in multiple rounds along each edge in the $\SIR$ model so that the $\SIR$ model and the $\IC$ model can be related by~\Cref{eqn:IC-SIR}.
However, if we view the $\SIR$ model in this way, infections across different outgoing edges of a node are correlated.
In this section, we prove that this correlation \emph{negatively} affects the cascade: given a graph $G$ and a set of seeds $S$ and setting the parameters of $\IC$ and $\SIR$ to satisfy \Cref{eqn:IC-SIR}, the influence spread of $\IC$ dominants $\SIR$. Moreover, we further show that the differences between $\IC$ and $\SIR$ lead to different seeding strategies. 

We first show the positive correlation property with respect to the occurrence of the edges in the $\SIR$ model (\Cref{sect:positive_cor}). Together with the reverse reachable set characterization of influence spread, we demonstrate that the influence spread in the $\IC$ model dominates the one in the corresponding $\SIR$ model conditioned on the comparable spreading ability on each edge by a coupling between the reverse reachable set (\Cref{thm:influence-spread-dominance} in~\Cref{subsect:IC>SIR}). Furthermore, in certain scenarios, we find that the $\IC$ model can significantly dominate the $\SIR$ model. 
We delve deeper into these scenarios to claim that the influence of $\IC$ could significantly dominate $\SIR$, which further implies different seeding strategies.

\subsection{Positive Correlation in the $\SIR$ Model}
\label{sect:positive_cor}

According to the live-edge graph formulation of the $\SIR$ model, the positive correlation can be described as follows: the probability of an edge being included in the live-edge graph decreases if other edges in the underlying graph $G$ are not included.

\begin{restatable}{lemma}{NegativeCorrelatde}\label{lem:negative-correlated}
     Given a directed graph $G=(V,E)$ with the diffusion model $\textnormal{SIR}_{\boldsymbol{\beta},\boldsymbol{\gamma}}$, we have 
    \begin{align*}
        \Pr{e\in \+G_{\SIR} \mid E'\cap \+G_{\SIR}=\emptyset}\leq \Pr{e\in \+G_{\SIR}}
    \end{align*} for any $E'\subseteq E$ and $e\in E\setminus E'$.
\end{restatable}
A tiny example is helpful to understand \Cref{lem:negative-correlated}. Given $G=(V,E)$ with $V=\{ u,v,w \}$, $E=\{(u,v),(u,w)\}$, and letting $e = (u,v)$ and $E'=\{(u,w)\}$, 
\Cref{lem:negative-correlated} says that the event $(u,w)\notin\+G_{\SIR}$ makes the edge $(u,v)$ less likely to be live.
The proof of this lemma is deferred to~\Cref{sec:appendix}. 
Intuitively, in the above example, if knowing $(u,w)$ is not live, $u$ is more likely to recover at earlier rounds, which decreases the chance that $(u,v)$ is live.
In general, knowing a set of edges fails to be live makes an edge less likely to be live.

\begin{remark}
    The dependency exists only on the outgoing edges from the same vertex.
    It is obvious from the definition of $\SIR$ that the events $e_1\in\+G_\SIR$ and $e_2\in\+G_\SIR$ are independent if $e_1$ and $e_2$ are not outgoing edges of the same vertex.
\end{remark}

To show that $\IC$ dominates $\SIR$ given the same seed set and with parameters satisfying \Cref{eqn:IC-SIR}, a natural idea is to couple the two spreading processes such that the set of infected vertices in the $\SIR$ process is a subset of the set of infected vertices in the $\IC$ process.
However, in the following example, we demonstrate that such a coupling is unlikely to exist.

Consider a seed set with a single seed $s$ that has many outgoing neighbors. 
In the $\IC$ model, the events that $s$ successfully infects its neighbors are independent.
In the $\SIR$ model, due to \Cref{eqn:IC-SIR}, $s$ infects each neighbor with the same probability as in the $\IC$ model.
Thus, the expected number of infected neighbors is the same in both models due to the linearity of expectation.
Intuitively, due to the positive correlation, in the $\SIR$ model, the number of infected neighbors is more likely to be either very small or very large.
However, because of the same expectation, it is impossible to find a coupling such that the set of infected neighbors in $\SIR$ is a subset of the set of infected neighbors in $\IC$.
In fact, from this example, it is not even intuitively clear that $\IC$ is superior to $\SIR$.

In the next section, we will use a coupling on sampling the reverse-reachable sets of an arbitrary vertex $v$.
We will show that viewing the cascade in such a ``backward'' way helps us better understand the relationship between the two models.

\subsection{$\IC$ Dominating $\SIR$}
\label{subsect:IC>SIR}
 


Given a directed graph $G=(V,E)$, we rigorously prove that the influence spread of any seed set in the $\IC$ model dominates that of the $\SIR$ model, under the assumption that the parameters satisfy \Cref{eqn:IC-SIR}. 
\begin{theorem}\label{thm:influence-spread-dominance}
    Given any directed graph $G=(V,E)$ together with the diffusion models $\IC_{\boldsymbol{p}}$ and $\SIR_{\boldsymbol{\beta,\gamma}}$ where the parameters satisfy \Cref{eqn:IC-SIR}, we have for any set $S\subseteq V$,
    \begin{align*}
        \sigma_{{\IC}}(S)\geq  \sigma_{{\SIR}}(S).
    \end{align*}
\end{theorem}

According to the linearity of the expectation and the reverse reachable set characterization of the influence spread~\Cref{lem:reverse-reachable-set-characterization}, it suffices to show the following lemma.
\begin{lemma}\label{lem:dominance-coupling}
    Given any directed graph $G=(V,E)$ together with the diffusion models $\IC_{\boldsymbol{p}}$ and $\SIR_{\boldsymbol{\beta,\gamma}}$, where the parameters satisfy \Cref{eqn:IC-SIR}, we have for any seed set $S\subseteq V$ and any node $v\in V$,
    \begin{align*}
        \Pr{S\cap \RR_{\IC}(v)\neq \emptyset}\geq  \Pr{S\cap \RR_{\SIR}(v)\neq \emptyset}.
    \end{align*}
\end{lemma}

The lemma above straightforwardly implies~\Cref{thm:influence-spread-dominance}.

\begin{proof}[Proof of \Cref{thm:influence-spread-dominance}]
    According to \Cref{lem:reverse-reachable-set-characterization}, we have
    \begin{align*}
        \sigma_{\IC}(S)&=\sum\nolimits_{v\in V}  \Pr{S\cap \RR_{\IC}(v)\neq \emptyset}\geq \sum\nolimits_{v\in V} \Pr{S\cap \RR_{\SIR}(v)\neq \emptyset}= \sigma_{\SIR}(S)
    \end{align*} where the inequality follows from \Cref{lem:dominance-coupling}.
\end{proof}

It now remains to prove \Cref{lem:dominance-coupling}.

\subsubsection{Some Intuitions for the Correctness of \Cref{lem:dominance-coupling}.}
To prove~\Cref{lem:dominance-coupling}, considering an arbitrary fixed $v$, we define a coupling between $\+G_\IC$ and $\+G_\SIR$ such that $\RR_{\SIR}(v)\subseteq\RR_{\IC}(v)$. The existence of such a coupling immediately implies the lemma. It then remains to define such a coupling.
We first describe the high-level ideas as follows.

The edges in both $\+G_\IC$ and $\+G_\SIR$ are revealed on a need-to-know basis.  
Specifically, the edges are revealed in a reverse Breadth-First-Search process: 
let $U$ denote the set of vertices that can reach vertex $v$, initialized as $U=\set{v}$; in each iteration, for every vertex $x$ that has out-neighbors in $U$, we reveal the corresponding outgoing edges of $x$; if one of these outgoing edges is live, then $x$ is included in $U$.
{In addition, we can maintain an in-arborescence (a directed tree rooted at $v$ where each edge is from the child to the parent) containing the reverse reachable nodes of $v$}: if the node $x$ has at least one outgoing edge connecting to $U$ that is live, we pick an arbitrary such live edge and include it in the in-arborescence.
Note that, after this process, unrevealed edges have no effect on $\RR_\IC(v)$ or $\RR_\SIR(v)$.
We couple the two reverse Breadth-First-Search processes for $\RR_\IC(v)$ and $\RR_\SIR(v)$ respectively such that the in-arborescence for the former is a superset of the in-arborescence for the latter.
The fact that each node in an in-arborescence has an out-degree at most $1$ ensures that this is always possible.
This is better illustrated by the following example.

Suppose at a certain stage of the Breadth-First-Search process where, in both $\IC$ and $\SIR$ processes, $U$ is the set of vertices that are already in the in-arborescence, and suppose now we are revealing the outgoing edges of a vertex $x\notin U$ to see if $x$ is in the in-arborescence.
Suppose, for example, $y_1,y_2,y_3\in U$ are all the outgoing neighbors of $x$ that belong to $U$.
If at least one of $(x,y_1)$, $(x,y_2)$, and $(x,y_3)$ is live in the $\SIR$ process, then $x$ is included in $U$ in the next iteration.
Our coupling ensures that $x$ is also included in $U$ in the $\IC$ process.
To see this, taking an example where $(x,y_2),(x,y_3)\in\+G_\SIR$ and $(x,y_1)\notin\+G_\SIR$, which happens with probability 
$\Pr{(x,y_1)\notin\+G_\SIR}\cdot\Pr{(x,y_2)\in\+G_\SIR\mid (x,y_1)\notin\+G_\SIR}\cdot\Pr{(x,y_3)\in\+G_\SIR\mid (x,y_1)\notin\+G_\SIR,(x,y_2)\in\+G_\SIR}$.
The live edge with the smallest index, which is $(x,y_2)$ in this case, is included in the in-arborescence, and $x$ is included in $U$.
We consider these three probabilities in the $\IC$ model.
For the first, we have $\Pr{(x,y_1)\in\+G_\SIR}=\Pr{(x,y_1)\in\+G_\IC}$ due to \Cref{eqn:IC-SIR}.
For the second, we have
\begin{align*}
    \Pr{(x,y_2)\in\+G_\SIR\mid (x,y_1)\notin\+G_\SIR}&\leq\Pr{(x,y_2)\in\+G_\SIR}\tag{\Cref{lem:negative-correlated}}\\
    &=\Pr{(x,y_2)\in\+G_\IC}\tag{\Cref{eqn:IC-SIR}}\\
    &=\Pr{(x,y_2)\in\+G_\IC\mid (x,y_1)\notin\+G_\IC}.\tag{Independence in $\IC$}
\end{align*}
By coupling the events for the first two edges $(x,y_1)$ and $(x,y_2)$, we already ensure that $x$ is included in $U$ in the $\IC$ process as well.
Notice that, for the third probability, although the relationship between $$\Pr{(x,y_3)\in\+G_\SIR\mid (x,y_1)\notin\+G_\SIR,(x,y_2)\in\+G_\SIR}$$ and $$\Pr{(x,y_3)\in\+G_\IC\mid (x,y_1)\notin\+G_\IC,(x,y_2)\in\+G_\IC}$$ cannot be implied by \Cref{lem:negative-correlated}, we have already included $x$ in $U$ regardless of the status of $(x,y_3)$.
Even if the conditional probability that $(x,y_3)$ is live for $\SIR$ is higher than it is for $\IC$, this does not give $\SIR$ any advantages.

More generally, if $x$ has many out-neighbors $y_1,y_2,\ldots,y_k$ in $U$ and $i$ is the smallest index such that $(x,y_i)$ is live in the $\SIR$ process, by \Cref{lem:negative-correlated}, the event that $(x,y_1),\ldots,(x,y_{i-1})$ are blocked reduces the chance that $(x,y_i)$ is live in the $\SIR$ process, while this event has no effect on the chance that $(x,y_i)$ is live in the $\IC$ process.
By coupling the event corresponding to the status of the first $i$ edges, we can make $x$ be included in $U$ for the $\IC$ process whenever this happens in the $\SIR$ process.
Although $(x,y_i)$ being live may increase the chances of $(x,y_{i+1}),\ldots,(x,y_k)$ being live in the $\SIR$ process, this does not give $\SIR$ any advantages against $\IC$ as $x$ is already included in $U$ due to the inclusion of $(x,y_i)$ in the in-arborescence.
The fact that we only need one live edge from $x$ to $U$ and the correlation property described in \Cref{lem:negative-correlated} make $\IC$ superior to $\SIR$.

Notice that the main purpose of the example above is to give readers some intuitions.
A rigorous proof follows next.

\subsubsection{The Proof of~\Cref{lem:dominance-coupling}.}

Given any node $v\in V$ with the live-edge graph $\+G$, we can determine the set $\RR_{\+G}(v)$ in a reverse Breath-First-Search way, and the edges are revealed on a need-to-know basis. Specifically, let $U$ denote the set of vertices that can reach vertex $v$, initialized as $U=\set{v}$ and $A\subseteq V\setminus U$ be the set of nodes that has out-neighbors in $U$. We repeat the following operations until there are no more {nodes that can be added into $U$,} i.e., $A=\emptyset$: select an arbitrary node $u\in A$, and reveal its outgoing edges to $U$ in an arbitrary fixed order. Once there exists an outgoing edge occurs in $\+G$, we update $U\gets U\cup \set{u}$ and skip the remaining unrevealed edges as these edges will not affect the reverse reachable set.
One can easily verify that $U=\RR_{\+G}(v)$ upon the termination of the process.

To show that $\IC$ dominates $\SIR$, we design a coupling between the reverse reachable set realization processes of $\IC$ and $\SIR$. As mentioned in~\Cref{subsect:IC>SIR}, the coupling relies on the in-arborescence structure in the SIR process. This is because the probability of an outgoing edge occurring in $\+G_{\SIR}$ is less than that in $\+G_{\IC}$ given that it is the first live edge revealed (according to~\Cref{lem:negative-correlated}).

Based on the need-to-know principle, we maintain an in-arborescence in the SIR process and a live-edge sub-graph in $\IC$ such that the former is the subset of the latter, which implies that the reverse reachable set of $\IC$ is a superset of the corresponding set in $\SIR$ at each step.
Specifically, let $E_1$ and $E_2$ be the random live edges in ${\IC}$ and ${\SIR}$ that have been revealed at each step. We also use $G_1=(V,E_1)$ and $G_2=(V,E_2)$ to denote the sub-graphs that consist of edges and vertices related to $E_1$ and $E_2$, respectively. Moreover, the set $\RR_{1}$ and $\RR_{2}$ correspond to the reverse reachable nodes of $v$ in $G_1=(V,E_1)$ and $G_2=(V,E_2)$, respectively. As one will see, $G_2$ forms an in-arborescence, and $G_2$ is always a subgraph of $G_1$, which implies that $\RR_{2}\subseteq \RR_{1}$ at each step.  

Let $A$ be the collection of nodes that have unrevealed outgoing edges pointing to $\RR_2$. At each round, we select a node $u\in A$ and reveal the unrevealed edges from $u$ to $\RR_2$.
We use $E'_u$ to denote the set of outgoing edges of the node $u$ that have been revealed in previous rounds, and  $\bigl((u,u_1),\ldots, (u,u_\ell) \bigr)$ to denote the unrevealed outgoing edges of $u$ with endpoints in $\RR_{2}$ listed by the index order. We reveal the edges in $\bigl((u,u_1),\ldots, (u,u_\ell) \bigr)$ in order and stop revealing more edges once an edge is live. During the coupling process, {for each unsampled edge $e=(u,u_j)$ where $j\in[\ell]$, we operates as follows:
\begin{enumerate}
    \item Sampling $p\sim\!{Unif}[0,1]$;
    \item For $\IC$, if $p\leq p_{u,u_j}$ then make $e$ live in the live-edge graph $\+G_{\IC}$, $\RR_{1}\gets \RR_{1}\cup \set{u}$, $E_1\gets E_1\cup \set{(u,u_j)}$;
    \item For $\SIR$, if $p\leq \Pr{(u,u_j)\in \+G_{\SIR}\mid E'_u\cap \+G_{\SIR}=\emptyset}$ then make $e$ live in the live-edge graph $\+G_{\SIR}$, $\RR_{2}\gets \RR_{2}\cup \set{u}$, $E_2\gets E_2\cup \set{(u,u_j)}$, skip the examination of the remaining edges;
    \item Setting $E'_u\gets E'_u\cup\set{e} $.
\end{enumerate}
Once there are no more nodes that have out-neighbors in $\RR_2$, we reveal the whole live-edge graph on the unrevealed edge. Specifically, let $\+G_1$ and $\+G_2$ be the random live-edge graph by revealing the edges in $E\setminus \left(\bigcup_{u\in V} E'_{u}\right)$ according to the correct marginal distributions, respectively. The formal description of our coupling is presented in~\Cref{algo:coupling}.} Note that the variable $\!{flag}$ in~\Cref{algo:coupling} is the indicator of whether there is an outgoing edge that is live.

\begin{algorithm}[h]
		\caption{The reverse recursive coupling $\+C_v$}
		\label{algo:coupling}
		\KwIn{A directed graph $G=(V,E)$ with the parameters $\boldsymbol{p}$ and $\boldsymbol{\beta,\gamma}$, and a node $v\in V$.}
		\KwOut{A pair of live-edge graph $\+G_1,\+G_2$.} 
        Label an arbitrary order on the nodes and edges in $G$\;
        $E_1=E_2\gets \set{\emptyset}$ and $\RR_{1}=\RR_{2}\gets \set{v}$\;
	$E'_u \gets \emptyset$ for any $u\in V$\;
                \While{$A\gets \set{u\in V\setminus \RR_{2}\mid \exists u'\in \RR_{2}, (u,u')\in E\setminus E'_u }$ and $A\neq \emptyset$}
                {
                    Select the node $u\in A$ with the smallest index\;
                    Let $\bigl((u,u_1),\ldots, (u,u_\ell) \bigr)$ be the directed edges not in $E'_u$  with endpoints in $\RR_{2}$ listed by the index order\;
                    $j\gets 0$, $\!{flag}\gets 0$\;
                    \Repeat{$\!{flag}=1$ or $j=\ell$}
                    {
                        $j\gets j+1$\;
                        Sample $p\sim\!{Unif}[0,1]$\;
                        $\RR_{1}\gets \RR_{1}\cup \set{u}$, $E_1\gets E_1\cup \set{(u,u_j)}$ when $p\leq p_{u,u_j}$\;
                   
                        $\RR_{2}\gets \RR_{2}\cup \set{u}$, $E_2\gets E_2\cup \set{(u,u_j)}$ and $\!{flag}\gets 1$ when $p\leq \Pr{(u,u_j)\in \+G_{\SIR}\mid E'_u\cap \+G_{\SIR}=\emptyset}$\;
                        $E'_u\gets E'_u\cup\set{(u,u_j)} $\;
                    }
                }
        Let $\+G_1$ and $\+G_2$ be the random live-edge graph by revealing the edges in $E\setminus \left(\bigcup_{u\in V} E'_{u}\right)$ according to the correct marginal distributions, respectively\;
        \Return{$\left(\+G_{1},\+G_{2}\right)$.}
\end{algorithm}


The realization of the edges follows from the correct marginal distribution, which ensures the correctness of our live-edge graphs.
\begin{restatable}{lemma}{correctdistribution}\label{lem: correctdistribution}
    In~\Cref{algo:coupling}, $\+G_{1}$ and $\+G_{2}$ follow the distribution of $\+G_{\IC}$ and $\+G_{\SIR}$, respectively.
\end{restatable}
\begin{proof}
    It is obvious that $\+G_{1}$ follows the distribution of $\+G_{\IC}$ since each edge $e\in E$ is included in $\+G_{\IC}$ independently. We next show the same fact holds for $\+G_{2}$. Let $E'=\bigcup_{u\in V} E_u'$. It suffices to show that our coupling process maintains $(E_2,E')$ with probability $\Pr{\mathcal{E}}$ where
    $\mathcal{E}$ denote the event that $E_{2}\subseteq \+G_{\SIR}$ and $(E'\setminus E_{2})\cap\+G_{\SIR}=\emptyset$.
    We prove it by a structural induction with the satisfying induction basis $E_2=E'=\emptyset$. For the induction step, we assume that $e=(u,u')$ is the next edge to be sampled given the current $E_2$ and $E'$. According to the coupling procedure, $e$ is included with probability $ \Pr{e\in \+G_{\SIR}\mid E'_u\cap \+G_{\SIR}=\emptyset}$. By definition, we have 
    \[
        \Pr{e\in \+G_{\SIR}\mid E'_u\cap \+G_{\SIR}=\emptyset}=  \Pr{e\in \+G_{\SIR} \mid \mathcal{E}}.
    \]
    Combining with the induction hypothesis that $(E_2,E')$ are obtained with $\Pr{\mathcal{E}}$, we add $e$ to $E_{2}$ and remove $e$ from $E'$ with probability
    \begin{align*}
        &\quad\mathbf{Pr}[e\in \+G_{\SIR} \,\vert\, \mathcal{E}]\cdot \mathbf{Pr}[\mathcal{E}]
        =\mathbf{Pr}[(E_{2}\! \cup  \set{e})\subseteq \+G_{\SIR},(E'\!\setminus E_{2})\cap\+G_{\SIR}=\emptyset].
    \end{align*} The case when $e$ is not included in the live-edge graph is similar, and the proof is complete.
\end{proof}

By sharing the randomness on the realization of each edge in $\IC$ and $\SIR$, we couple the propagation of the influence and maintain that $ \RR_{1}\supseteq \RR_{2}$ by the positive correlation property in the $\SIR$ model. 
\begin{restatable}{lemma}{RRsetdomainance}\label{lem:RRset-domainance}
    Given any directed graph $G=(V,E)$ together with the diffusion models $\IC_{\boldsymbol{p}}$ and $\SIR_{\boldsymbol{\beta,\gamma}}$, we have that 
    \begin{align*}
    \forall\, v\in V, \quad
        \RR_{\+G_1}(v)\supseteq \RR_{1}\supseteq \RR_{2}=\RR_{\+G_2}(v)
    \end{align*}
     when $\Pr{e\in \+G_{\IC}}=\Pr{e\in \+G_{\SIR}}$ for any $e\in E$.
\end{restatable}

\begin{proof}
     We first claim that $\RR_{2}=\RR_{\+G_2}(v)$. Since $\RR_{2}\subseteq \RR_{\+G_2}(v)$, it suffices to show $\RR_{\+G_2}(v)\setminus \RR_{2}=\emptyset$. Otherwise, for arbitrary $u\in \RR_{\+G_2}(v)\setminus \RR_{2}$, there exists a directed path $P=(u_1=u,u_2,\dots,u_\ell=v)$ in $\+G_2$. However, our coupling procedure would include each $u_i$ ($i\in [\ell]$) in $\RR_{2}$ contradicting $\RR_{\+G_2}(v)\setminus \RR_{2}\neq \emptyset$.
     
     As for $\RR_{1}\supseteq \RR_{2}$, it holds simply by the positive correlation property of $\SIR$ in \Cref{lem:negative-correlated} and the assumption $\Pr{e\in \+G_{\IC}}=\Pr{e\in \+G_{\SIR}}$ for any directed edge $e\in E$. Again, $\RR_{\+G_1}(v)\supseteq \RR_{1}$ is obvious, and the proof is complete combining all these facts.
 \end{proof}
 \vspace{-1mm}

     

We are now ready to complete the proof of \Cref{lem:dominance-coupling}.

\begin{proof}[Proof of \Cref{lem:dominance-coupling}]
    As~\Cref{lem: correctdistribution} shown, $(\+G_{1},\+G_{2})$ is the coupling of $\+G_{\IC}$ and $\+G_{\SIR}$. Moreover, for any seed set $S\subseteq V$, $S\cap \RR_{\SIR}(v)\neq \emptyset$ implies $S\cap \RR_{\IC}(v)\neq \emptyset$ since $\RR_{\+G_1}(v)$ is the superset of $\RR_{\+G_2}(v)$ by \Cref{lem:RRset-domainance}. Hence, the inequality holds.
\end{proof}

\subsection{Significant Dominance}
In the previous section, we have shown that, conditioned on the comparable spreading ability of each edge, there is a domination from $\IC$ to $\SIR$. In this section, we demonstrate that such domination can be significant in some networks. This dominance further yields significantly different seeding strategies under these two models. The graph we construct is represented in~\Cref{fig:gadget}. 

\begin{figure}[!h]
    \centering
    \includegraphics[scale=0.4]{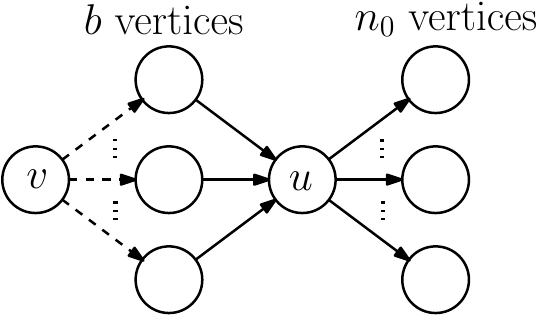}
    \caption{An instance to show the significant dominance. A vertex $v$ is connected to $b$ vertices by $b$ dashed edges, and these $b$ vertices are connected to a vertex $u$ by $b$ solid edges. Finally, $u$ is connected to $n_0$ vertices by $n_0$ solid edges. The solid edges are ``deterministic'', meaning that the parameters $\beta$ and $p$ in both $\IC$ and $\SIR$ models are set to $1$. For the dashed edges, they share the same $\SIR$ parameters $\beta$ and $\gamma$, and their $\IC$ parameter is decided based on $\beta$, $\gamma$ and \Cref{eqn:IC-SIR}. The values of $b$, $n_0$, $\beta$, and $\gamma$ are to be decided.}
    \label{fig:gadget}
\end{figure}


We show that for some specific social networks, the influence spread of $\IC$ could significantly dominate $\SIR$ with parameters satisfying \Cref{eqn:IC-SIR}.

The following proposition shows that, for the instance in \Cref{fig:gadget}, the influence spread of $\IC$ can significantly dominate that of $\SIR$ for the seed set $\{v\}$ if the parameters $b$, $n_0$, $\beta$, and $\gamma$ are set appropriately.
\begin{restatable}{proposition}{ICDominateSIRGreatly}\label{prop:greatdominate}
    Consider the instance in \Cref{fig:gadget}. For any $R>0$, there exist $b$, $n_0$, $\beta$, and $\gamma$ such that $\frac{\sigma_\IC(\{v\})}{\sigma_\SIR(\{v\})}>R$.
\end{restatable}

The high-level idea for \Cref{prop:greatdominate} is as follows.
By setting $n_0$ to be very large, $\sigma(\{v\})$ under both models is approximately proportional to the probability that vertex $u$ is infected.
It then remains to argue that $u$ is infected with a significantly smaller probability in the $\SIR$ model.
Notice that, by the definition of the solid edges, $u$ is infected (with probability $1$) once one of its in-neighbors is infected.
Therefore, all we care about is the probability that \emph{at least one of the $b$ vertices is infected}, which is $1$ minus the probability that none of the $b$ vertices is infected.
The latter probability is significantly higher in the $\SIR$ model due to the positive correlation in \Cref{lem:negative-correlated}.
The formal proof of \Cref{prop:greatdominate} is available in \Cref{sec:appendix}.


\subsection{Difference in Seeding Strategies}
We have shown that the expected number of infections under $\IC$ is weakly larger than that under $\SIR$ given the same seed set and with the matching parameters, and sometimes $\IC$ significantly dominates $\SIR$.
This significant dominance yields significant different seeding strategies under these two models, as shown by an example illustrated in~\Cref{fig:gadget2}.

In~\Cref{fig:gadget2}, the two red nodes are potential seeds.
The red node on the left-hand side is the center of a star, while, on the right-hand side, we duplicate the gadget in~\Cref{fig:gadget} such that the node $v$ in~\Cref{fig:gadget} is the red node (i.e., all the duplications share only the same vertex $v$; two duplications are drawn in \Cref{fig:gadget2}).
By \Cref{prop:greatdominate}, the seed on the right-hand side is much more powerful in the $\IC$ model than that in the $\SIR$ model; the seed on the left-hand side has the same power in both models.
The reader can easily verify that, by appropriately setting up the parameters, the seed on the left-hand side is a much better choice for the $\SIR$ model and the seed on the right-hand side is a much better choice for the $\IC$ model.

\begin{figure}[!h]
    \centering
    \includegraphics[scale=0.35]{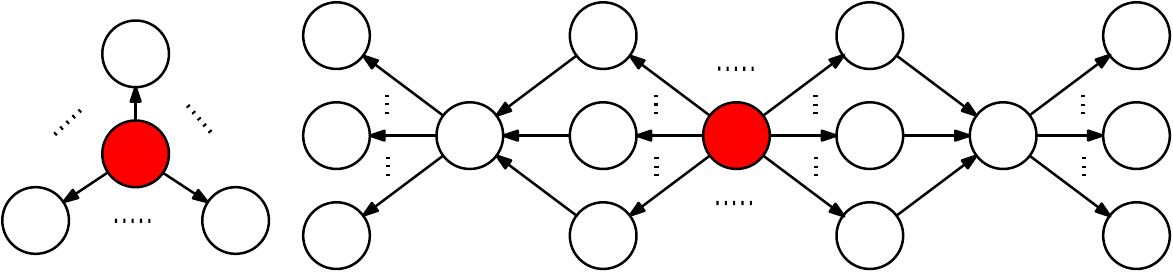}
    \caption{A graph to show the discrepancy of seeding strategy in $\IC$ and $\SIR$.
    \label{fig:gadget2}}
\end{figure}

\section{$\InfMax$ in the $\SIR$-based Models}
\label{sec:SIRAlgorithms}
In this section, we present algorithmic results for $\InfMax$ in the $\SIR$ model. 
It is worth noting that $\InfMax$ in various diffusion models can be formulated as a \emph{submodular} optimization problem.
In our case with the $\SIR$ model, the expected influence spread function $\sigma(\cdot)$ is also submodular.
This follows from the fact that a diffusion model with a live-edge graph formulation is always submodular~\citep{KKT15}.
With submodularity, it is well-known that the simple greedy algorithm achieves a $(1-1/e)$-approximation~\citep{nemhauser1978analysis}.
However, computing the function $\sigma(\cdot)$ is a \#P-hard problem~\citep{KKT15,CCY10,CYZ10}.
To apply the greedy algorithm, the influence spread function $\sigma(\cdot)$ is approximated by the Monte-Carlo method, so we have a randomized algorithm that achieves a $(1-1/e-\varepsilon)$-approximation with high probability.
Notice that this approximation ratio is optimal assuming $\text{P}\neq\text{NP}$, as $\InfMax$ is NP-hard to approximate to within a factor of more than $(1-1/e)$ for the $\IC$ model~\citep{KKT15}, and the special case of the $\SIR$ model with $\gamma_u=1$ for every $u\in V$ is exactly the $\IC$ model.

Despite that the greedy algorithms based on the Monte-Carlo are slow in practice.
Researchers have then focused on designing algorithms that are fast in practice while keeping the theoretical approximation guarantee.
One of the most successful types of such algorithms is based on \emph{reverse-reachable sets} \citep{BBCL2014,TYXY14,TSX2015,CW18}, which yields a near-linear time algorithm with $1-1/e$ approximation ratio.
The reverse reachable set technique, proposed by~\cite{BBCL2014}, makes use of the observation in \Cref{lem:reverse-reachable-set-characterization}.
Our objective is to find $S$ that maximizes $\Pr{{S\cap \RR_{\+G}}\neq \emptyset}$.
This is done by sampling a sufficient number of copies of $\RR_\+G$ and finding $S$ that intersects as many copies as possible, which becomes a classical $k$-max coverage problem.
This type of algorithm has been widely used for the $\IC$ model and proved successful in terms of scalability.

Subsequent reverse-reachable-set-based algorithms aim to reduce the number of $\RR_{\+G}$ sets sampled while still keeping the approximation guarantee.
Notably, the $\IMM$ algorithm proposed by \cite{TSX2015} uses a martingale method to estimate the number of reverse reachable sets. 

\begin{algorithm}[!h]
		\caption{The $\IMM$ algorithmic framework $\IMM(G,k,\epsilon,\ell)$}
		\label{algo:IMM}
		\KwIn{A directed graph $G=(V,E)$ where $\vert{V}\vert=n,\vert{E}\vert=m$ with live-edge diffusion model $\+G$, the size of the seed set $k$ and parameters $\epsilon,\ell$.}
		\KwOut{The seed set $S^*$.}
		$\alpha \gets \sqrt{\ell \cdot \log n+\log 2}$, $\beta \gets \sqrt{\left(1-1/e\right) \cdot \left(\log \left(  _k^n \right)+\ell \cdot \log n+\log 2\right)}$ and $\gamma \gets 4 + \log (8 \cdot \log n) / \log n$\;
		$\ell' \gets \ell +\log 2 /\log n + \gamma$ and $\epsilon'\gets \sqrt{2}\cdot \epsilon$\; $\lambda'\gets \frac{(2+\frac{2}{3}\epsilon')\cdot \big(  \log\binom{n}{k} +  \ell'\cdot \log n + \log \log_2 n\big)\cdot n}{\epsilon'^2}$ and $\lambda^* \gets 2 \cdot n \cdot \left( \left( 1-1/e  \right) \cdot \alpha + \beta \right)^2 \cdot \varepsilon ^{-2} $\;	
        Initialize a set $\+R=\emptyset$ and an integer $\~{LB}=1$\;
        \For{$i=1$ to $\log_2 n-1$}
        {
            $x\gets \frac{n}{2^i}$ and $\theta_i=\frac{\lambda'}{x}$\;
            \While{$|\+R|\leq \theta_i$}
            {
                $\RR\gets \!{RRSampling}(G,\+G)$\tcp*{sample a single $\RR_{\+G}$ set}
               $\+R\gets \+R\cup \set{\RR}$\;
            }
            $S_i\gets \!{NodeSelection}(G,\+R,k)$\tcp*{see \Cref{algo:nodeselection}}
            \If{$n\cdot F_{\+R}(S_i)\geq (1+\epsilon')\cdot x $}
            {
               $\~{LB}=n\cdot F_{\+R}(S_i)/(1+\epsilon')$\;
               \textbf{break}\;
            }
        }
        $\theta=\lambda^*/\~{LB}$\;
        \While{$|\+R|\leq \theta$}
            {
                $\RR\gets \!{RRSampling}(G,\+G)$\tcp*{sample a single $\RR_{\+G}$ set}
               $\+R\gets \+R\cup \set{\RR}$\;
            }
         $S^*\gets \!{NodeSelection}(G,\+R,k)$\tcp*{see \Cref{algo:nodeselection}}
        \Return{$S^*$.}
\end{algorithm}
\begin{algorithm}[!h]
		\caption{$\NODE(G,\+R,k)$}
		\label{algo:nodeselection}
		\KwIn{A directed graph $G=(V,E)$, a collection of $\RR$ set $\+R$ and the size of the seed set $k$.}
		\KwOut{A seed set $S^*$.}

        Initialize $S^*\gets \emptyset$\;
        \While{$(|S^*|<k)$ and $(V\setminus S^*\neq \emptyset)$}
        {
            $v\gets \argmax_{u\in V} F_{\+R}(S^*\cup u)-F_{\+R}(S^*)$\;
            $S^*\gets S^*\cup \set{v}$\;
        }
        \Return{$S^*$.}
\end{algorithm}

Compared to the algorithm proposed by~\cite{BBCL2014}, the $\IMM$ algorithm leverages the advantages of optimizing the size of the reverse reachable set required for estimating the influence spread.
The detailed implementation is shown in~\Cref{algo:IMM}.\footnote{An issue in the original IMM algorithm was later discovered by~\citet{CW18}. The algorithm we present here is based on the version after the correction in~\citet{CW18}.}
In the algorithm, we use $\+R$ to collect all the sampled reverse reachable sets and use $F_{\+R}(S)$ to denote the fraction of reverse reachable sets in $\+R$ covered by any seed set $S\subseteq V$. Here, the framework fits many live-edge graph diffusion models by replacing the reverse reachable set sampling subroutine $\!{RRSampling}(G,\+G)$. 

Let $R\subseteq V$ be the vertices of a reverse reachable set.
Let $w(R)$ be the number of edges of $G$ incident to $R$.
Examining the analysis of the $\!{IMM}$ algorithm in~\citet{TSX2015} and \citet{CW18}, one can verify the following lemma.

 \begin{lemma}\label{lem-correctness-IMM}
     Consider a directed graph $G=(V,E)$ where $|V|=n, |E|=m$ with live-edge diffusion model $\+G$, the size of the seed set $k$ and parameters $\varepsilon,\ell$. If the subroutine $\!{RRSampling}(G,\+G)$ returns an $\RR_{\+G}$ set $R$ in $\alpha\cdot w(R)$ time, the algorithm $\IMM(G,k,\varepsilon,\ell)$ (\Cref{algo:IMM}) returns a $(1-1/e-\varepsilon)$-approximate solution with at least $1-1/n^{\ell}$ probability in $O(\alpha\cdot(k+\ell)(m+n)\log n/\epsilon^2)$ time.
\end{lemma}

In the following, we introduce how we adapt the $\IMM$ algorithm to the $\SIR$ model and its variants $\TSIR$.
Specifically, we will design the subroutine $\!{RRSampling}$ that is specifically for the $\SIR$ model and the $\TSIR$ model respectively.

\subsection{Algorithm for \InfMax with $\SIR$ Model}\label{sec:SIRIMM}

For the $\SIR$ model, one natural idea is to sample each $\RR_\SIR$ based on the live-edge graph definition in \Cref{subsec:live-edge-SIR}. However, both $\boldsymbol{R}_v$ and $\boldsymbol{I}_e$ are sequences of infinitely many random variables.
To accurately and efficiently sample a reverse reachable set $\RR_\SIR$, we instead directly calculate the probability that an edge is live.
For example, $\Pr{e\in\+G_\SIR}$ can be computed by \Cref{eqn:IC-SIR}, which is a geometric series that admits a compact formula.
More generally, for a vertex $u$ with a set $E'_u$ of $u$'s incident edges and a particular incident edge $e$, the conditional probability $\Pr{e\in\+G_\SIR\mid E'_u\cap\+G_\SIR=\emptyset}$ can also be expressed as a geometric series by applying the law of conditional probability.
With these, the reverse reachable set $\RR_{\SIR}$ can be sampled in a way similar to the coupling procedure~(\Cref{algo:coupling}) in the proof of \Cref{lem:dominance-coupling}.
Details for our algorithms are available in \Cref{algo:SIRRRset}.
\begin{algorithm}[!h]
		\caption{$\RR_{\SIR}\!{Sampling}(G,\boldsymbol{\beta},\boldsymbol{\gamma})$}
		\label{algo:SIRRRset}
 		\KwIn{A directed graph $G=(V,E)$ with the $\SIR_{\boldsymbol{\beta},\boldsymbol{\gamma}}$ diffusion model.}
		\KwOut{A $\RR_{\SIR}$ set.}
        Label an arbitrary order on the nodes and edges in $G$\;
        Sample a node $v\in V$ uniformly at random\;
        Initialize $\RR_{\SIR} \gets \set{v}$, and $E'_u \gets \emptyset$ for any $u\in V$\;
		\While{ $A\gets \set{u\in V\setminus \RR_{\SIR}\mid \exists u'\in \RR_{\SIR}, (u,u')\in E\setminus E'_u }$ and $A\neq \emptyset$}
                {
                    Select the node $u\in A$ with the smallest index\;
                    Let $\Big((u,u_1),(u,u_2),\dots, (u,u_\ell) \Big)$ be the directed edges not in $E'_u$ with endpoints in $\RR_{\SIR}$ listed by the index order\;
                    $j\gets 0$, $\!{flag}\gets 0$\;
                    \Repeat{ $j=\ell$ or $\!{flag}= 1$}
                    {
                        $j\gets j+1$\;
                        Sample $p\sim\!{Unif}[0,1]$\;
                        $\RR_{\SIR}\gets \RR_{\SIR}\cup \set{u}$, $\!{flag}\gets 1$ when $p\leq \Pr{(u,u_j)\in \+G_{\SIR}\mid E'_u\cap \+G_{\SIR}=\emptyset}$\;
                        
                         $E'_u\gets E'_u\cup\set{(u,u_j)} $\;
                    }
                }
        \Return{$\RR_{\SIR}$.}
\end{algorithm}
Based on the $\IMM$ algorithmic framework, we propose the $\SIRIMM$ method by extending the $\IMM$ algorithm to $\SIR$, specifically, we designed the $\RR_{\SIR}\!{Sampling}(G,\boldsymbol{\beta},\boldsymbol{\gamma})$ that sample the reverse reachable set for $\SIR$. The sampling algorithm can be easily adapted from the coupling process ~\Cref{algo:coupling}. As mentioned in~\Cref{sec:SIRAlgorithms}, we can directly compute the probability that an edge is live. 
Specifically, for each node $u$, let $\Gamma_u = \{v_1,\dots,v_\ell\}$ represent the set of outgoing neighbors of $u$ in a fixed order. Denote by $\gamma$ and $\beta$ the recovery probability and infection probability in the $\SIR$ model, respectively. We compute the probability that $u$ successfully infects its $i$-th neighbor, conditioned on $u$ having failed to infect all of the first $i-1$ neighbors, as follows:
\begin{align*}
    &\Pr{v_i \mid \overline{ v_{[i-1]}}} \\
    =& \frac{\Pr{v_i \wedge \overline{ v_{[i-1]}}}}{\Pr{\overline{ v_{[i-1]}}}}\\
    =&\frac{\sum_{t=1}^\infty (1-\gamma)^{(t-1)} \cdot \gamma \cdot (1-\beta)^{(i-1) \cdot t}(1-(1-\beta)^t)}{\sum_{t=1}^{\infty}(1-\gamma)^{(t-1)}\cdot \gamma \cdot (1-\beta)^{(i-1)\cdot t}}\\
    =&\frac{\frac{(1-\beta)^{(i-1)}}{1-(1-\gamma)(1-\beta)^{(i-1)}}-\frac{(1-\beta)^i}{1-(1-\gamma)(1-\beta)^{i}}}{\frac{(1-\beta)^{(i-1)}}{1-(1-\gamma)(1-\beta)^{(i-1)}}}
\end{align*}
where $v_i \in \Gamma_u$ is the $i$-th neighbor of $u$, and $v_{[i-1]} \subseteq \Gamma_u$ represents the first $i-1$ neighbors of $u$, and we slightly abuse the notations by letting $v_i$ denote the \emph{event} that $u$ successfully infects $v_i$ and $\overline{v_{[i-1]}}$ denote the \emph{event} that $u$ fails to infect each vertex in $v_{[i-1]}$.

Therefore, the probability $\Pr{v_i \mid \overline{ v_{[i-1]}}}$ can be computed in $O(1)$ time.
With this, the remaining details for sampling a reverse-reachable set are the same as they are in the $\IC$ model.
Our algorithm is presented in Algorithm~\ref{algo:SIRRRset}.
According to~\Cref{lem:RRset-domainance}, our algorithm returns $\RR_{\SIR}$ faithfully.

Replacing the $\RR$ Sampling procedure in~\Cref{algo:IMM} by~\Cref{algo:SIRRRset}, we obtain the $\SIRIMM$ which returns an $(1-1/e-\epsilon)$-approximate solution with at least probability $1-1/n^\ell$.

Since the (conditional) probability that an edge is live can be computed in $O(1)$ time, the time complexity for sampling a single reverse reachable set is the same as in the $\IC$ model, namely, $\alpha\cdot w(R)$ for $\alpha=O(1)$.
By Lemma~\ref{lem-correctness-IMM}, the time complexity for $\SIRIMM$ is $O((k+\ell)(m+n)\log n/\epsilon^2)$.

\begin{theorem}
    There exists an algorithm that takes as inputs $k,\ell\in\mathbb{Z}^+$ and $\epsilon\in\mathbb{R}^+$, and a directed graph $G=(V,E)$ where $|V|=n, \vert E\vert=m$ with the diffusion model $\SIR_{\boldsymbol{\beta,\gamma}}$ and outputs a $(1-1/e-\epsilon)$-approximately optimal expected influence spread of $k$ seeds with at least $1-1/n^\ell$ probability in \(O\left((k+l) \cdot (n+m) \cdot \log n / \varepsilon^2\right)\) time.
\end{theorem}

\subsection{Algorithm for \InfMax with $\TSIR$ Model}

The case with the $\TSIR$ model is more involved due to the extra time-dependency.
Intuitively, $\RR_\TSIR(v)$ is a truncated version of $\RR_\SIR(v)$.
In addition, we need extra information in addition to the probability $\Pr{e=(u,v)\in\+G_\SIR}$, as it now matters in which round the infection across $(u,v)$ succeeds.
Therefore, we have to apply the original definition in \Cref{sec:LiveedgeTSIR} to sample a reverse reachable set.
Due to the time-dependency feature, we need to truncate those vertices in $\RR_\SIR(v)$ that are too far away from $v$.
This is done by applying Dijkstra's algorithm to maintain a reverse shortest path tree rooted at $v$.
Notice that we do not need to consider infinitely many random variables for $\boldsymbol{R}_v$ and $\boldsymbol{I}_e$: vertices that are not recovered or edges that are not active for long periods need not be considered due to the time-dependency of the $\TSIR$ model.

We design the approximation algorithm $\TSIRIMM$ for $\InfMax$ in the $\TSIR$ model by simulating the live-edge graph for the $\TSIR$ model to sample a reverse reachable set, as shown in~\Cref{algo:TSIRRRset}.The algorithm begins by selecting a node $v\in V$ uniformly at random, and $v$ is added to the set $\RR_\TSIR$.
Whenever a new node $u$ is added to $\RR_\TSIR$, we explore all its incoming neighbors.
Notice that those neighbors may have been explored before.
For each in-neighbor $w$, if it has not been explored before, we sample its recovery time; otherwise, the recovery time for $w$ has been sampled before, and we do nothing (Line 7~11).
After this, we sample the number of rounds for the infection along the edge $(u,w)$ to be successful (Line 14-22).
If the number of rounds is more than the recovery time of $w$, the edge $(u,w)$ is not included in the edge set $E_\TSIR$.
Otherwise, we include $(u,w)$ in $E_\TSIR$, and let the said number of rounds be its \emph{weight}.
We now end up with a weighted graph (Line 25), where vertices that are reachable to $v$ are those vertices that can infect $v$ \emph{without the time constraint $T$}.
Finally, to impose the time constraint $T$, we find all vertices $u$ whose distance to $v$ is at most $T$ (Line 26).
This can be done by the standard Dijkstra's algorithm.

\begin{algorithm}[!h]
        \SetAlgoLined
		\caption{$\RR_{\TSIR}\!{Sampling}(G,\boldsymbol{\beta},\boldsymbol{\gamma},T)$}
		\label{algo:TSIRRRset}
		\KwIn{A directed graph $G=(V,E)$ with parameters $\boldsymbol{\beta,\gamma}$ and a time limit $T$.}
		\KwOut{An $\RR_{\TSIR}$ set.}
        Sample a node $v\in V$ uniformly at random\;
        Initialize $\RR_{\TSIR} \gets \set{v}$, $E_{\TSIR}\gets \emptyset$, and $\~{RecRound}_u \gets \bot$ for each $u\in V$\;
        Initialize $\tt{ToBeProcessed}\gets \set{v}$\;
        \While{$\tt{ToBeProcessed}\neq\emptyset$}{
            Get $u\in\tt{ToBeProcessed}$ and remove $u$ from $\tt{ToBeProcessed}$\;
            \For{each in-neighbor $w$ of $u$}{
                $t\gets0$\;
                \While{$\~{RecRound}_w = \bot$ and $t\leq T$}
                {
                    $t\gets t+1$ and sample $p \sim\!{Unif}[0,1]$\;
                    set $\~{RecRound}_w\gets t$ if $p\leq \gamma_w$\;
                }
                $t\gets 0$\;
                \Repeat{$t>\~{RecRound}_w$ or $t>T$}
                {
                    $t\gets t+1$ and sample $p \sim\!{Unif}[0,1]$\;
                    \If{$p\leq\beta_{w,u}$}{
                        update ${\tt ToBeProcessed}\gets{\tt ToBeProcessed}\cup\{w\}$ if $w\notin\RR_\TSIR$\;
                        $E_{\TSIR}\gets E_{\TSIR}\cup \set{(w,u)}$\;
                        $\RR_{\TSIR} \gets \RR_{\TSIR}\cup \{w\}$\;
                        $w_{w,u}\gets t$\;
                        \textbf{break}\;
                    }
                }
            }
        }
        $G^* \gets$ the edge-weighted graph $(\RR_\TSIR,E_{\TSIR},\set{w_e}_{e\in E_{\TSIR}})$\;
		update $\RR_{\TSIR}\gets \set{u\in \RR_{\TSIR}\mid \mbox{the distance from $u$ to $v$ is at most $T$ in $G^*$}}$\;
        \Return{$\RR_{\TSIR}$.}
\end{algorithm}

Finally, to analyze the time complexity of Algorithm~\ref{algo:TSIRRRset}, notice that sampling the recovery time for each vertex and the number of rounds for which the infection along each edge $(w,u)$ is successful requires $O(T)$ time, as we only need to flip the corresponding coin for at most $T$ times.
In addition, we need to perform Dijkstra's algorithm at the end, which requires $O(|E_\TSIR|+|\RR_\TSIR|\cdot\log(|\RR_\TSIR|))$ time, and we take the upper bound $O(w(R)\log(n))$.
Therefore, the time complexity for sampling a reverse reachable set is $O(T\log(n)\cdot w(R))$.
By Lemma~\ref{lem-correctness-IMM}, the overall time complexity for $\TSIRIMM$ is $O(T\cdot(k+l)\cdot(n+m)\cdot\log^2n/\epsilon^2)$.
\begin{theorem}
    There exists an algorithm that takes as inputs $k,\ell,T\in\mathbb{Z}^+$ and $\epsilon\in\mathbb{R}^+$, and a directed graph $G=(V,E)$ where $|V|=n, \vert E\vert=m$ with the diffusion model $\TSIR_{\boldsymbol{\beta,\gamma},T}$ and outputs a $(1-1/e-\epsilon)$-approximate optimal expected influence spread of $k$ seeds with at least $1-1/n^\ell$ probability in \(O\left(T\cdot (k+l) \cdot (n+m) \cdot \log^2 n / \varepsilon^2\right)\) time.
\end{theorem}

\section{Ablation Experiment}
\label{sec:experiments}

Our experiments consist of two parts. The first part compares the spreading influence of $S$ under $\IC$, $\SIR$, and $\TSIR$, with setting the parameter $p_{u,v}$ in $\IC$ and the infection rate $\beta_{u,v}$ and recover rate $\gamma_v$ in $\SIR$ and $\TSIR$ satisfying \Cref{eqn:IC-SIR}:
\begin{align*}
    p_{u,v} = \sum_{t=1}^{\infty} \gamma_u (1-\gamma_u)^{t-1} (1- (1-\beta_{u,v})^t).
\end{align*}

In the second part of our experiments, we compared the performance (including the seed qualities and the running times) of our proposed algorithm $\SIRIMM$ with four baseline methods, including $\CELF$, $\IMM$, $\DegDis$, and the degree centrality algorithm, under the $\SIR$ model.
The same experiments are performed for the $\TSIR$ model as well.

\subsection{The Baseline Methods}
The baseline methods we applied including $\IMM$~\citep{TSX2015}, $\CELF$~\citep{LKGFVG2007}, $\DegDis$~\citep{WYS09}, and, degree centrality algorithms \citep{Freeman1978}. In this subsection, we will provide detailed descriptions of each baseline method. The details of $\IMM$ have already been described in \Cref{sec:SIRAlgorithms}. 

\subsubsection{The Degree Discount IC}
The Degree Discount heuristic ($\DegDis$)~\citep{WYS09} is a popular algorithm used for identifying influential nodes in a network. The $\DegDis$ heuristic works by iteratively selecting the node with the highest degree and then removing it from the graph. This algorithm has been referred to as the ``single degree discount heuristic'' in previous literature, and we will adopt the name $\DegDis$ in this study. Another variant of this heuristic, known as ``Degree Discount IC" ($\DegDis$ IC), was specifically designed by Chen et al.~\citep{CPL2012} for the \emph{Uniform Independent Cascade model}, a special case of the $\IC$ model where the activation probability is the same for all edges. 

In $\DegDis$ IC, the score of a candidate seed is computed by subtracting the number of its neighbors that have already been selected as seeds from its degree. This approach discounts the edges connecting the candidate seed to the existing seeds since they do not play a role in further infections. However, when considering the uniform independent cascade model, especially with a small parameter $p$, discounting by $1$ can be inaccurate. To address this issue, the $\DegDis$ IC heuristic uses a more precise estimation of the \Cref{alg:DegreeDiscount}. For more details on this heuristic, readers can refer to reference~\citep{CPL2012}.

\begin{algorithm}[!h] 
\BlankLine
\KwIn{A directed graph $ G=(V,E) $ with live-edge diffusion model $\+G$, the number of spreaders $k$, and the parameter$ p$.}
\KwOut{A set $S$ including $l$ influential nodes.}
Initialize $S = \emptyset$; \\
\For {$v \in V$}
{
    $\~{dd}_v = d_v$;  // $d_v$ is the degree of $v$\\
    Initialize $t_v = 0$;\\
}
\For{$i=1$ to $k$}
{
    Select $u = \mathop{\arg\max}_{v \in V}\{ \~{dd}_v \, |\, v \in V\setminus S \}$;\\
    $S = S \cup {u}$;\\
    \For{each neighbor $v$ of $u$ and $v \in V \setminus S$}
    {
        $t_v = t_v + 1$;\\
        $\~{dd}_v = d_v - 2 \cdot t_v - (d_v - t_v)\cdot t_v\cdot p$;
    }
}
\Return{$ S $};
\caption{$\!{Degree Discount}\left( G, k\right)$}
\label{alg:DegreeDiscount}
\end{algorithm}

\subsubsection{The Degree Centrality}
The degree centrality method~\citep{Freeman1978} simply chooses $k$ nodes with the largest degrees as the seeds.

\subsubsection{$\CELF$}

$\CELF$ is an improved variant of the standard greedy algorithm that is 700 times faster~\citep{LKGFVG2007}.
In each iteration of the standard greedy algorithm, all nodes in the graph are tested and the node with the highest marginal influence is selected.
In each iteration of $\CELF$, by further exploiting submodularity of the cascade process, it identifies some nodes whose marginal influences are clearly sub-optimal and excludes these nodes from being tested.
The core observation is that the marginal influence of each node in the current iteration will only be smaller than or equal to it is in the previous iteration.
The algorithm's details are shown in \Cref{alg:CELF}.

\begin{algorithm}[!h] 
\BlankLine

$\mathcal{A} \gets \!{LazyForward} \left( \+G = \left(V,E \right), R, c, B, \~{UC} \right) $\;
$\mathcal{A} \gets \!{LazyForward} \left( \+G = \left(V,E \right), R, c, B, \~{CB} \right) $\;

\Return{$ \operatorname{R}\left( \mathcal{A}_{\~{UC}},\mathcal{A}_{\~{CB}} \right) $}
\caption{$\CELF$$\left( G = \left(V,E \right), R, c, B, type \right)$ }
\label{alg:CELF}
\end{algorithm}

\begin{algorithm}[!htb]
\small{
\BlankLine

$ \mathcal{A} \gets \emptyset $\;
\ForEach{$ s \in V $} 
{
    $\delta_s \gets +\infty$\;
	\While{$ \exists s \in V \setminus \mathcal{A}:c \left( \mathcal{A} \cup \{s\} \right)$}
     {
         \ForEach{$ s \in V \setminus \mathcal{A}  $}  {$cur_s \gets \!{False}$\;}
         \While{$\!{True}$}
         {
          \If{$type = \~{UC}$}{$ s^* \gets \mathop{\arg\max}\limits_{s \in V \setminus \mathcal{A},c \left( \mathcal{A} \cup \{s\} \right) \leq B}\delta_s $\;}

         \If{$type = \~{CB}$}{$ s^* \gets \mathop{\arg\max}\limits_{s \in V \setminus \mathcal{A},c \left( \mathcal{A} \cup \{s\} \right) \leq B}\delta_s / c(s)$\;}
         \If{$cur_S$}{$\mathcal{A}\gets \mathcal{A} \cup s^*$\;
         Break\;}
         \Else {$\delta_s \gets R \left( \mathcal{A} \cup \{s\} - R(\mathcal{A}) \right)$\;$cur_s \gets \!{True}$\; }
         }
     }
}
\Return{$ \mathcal{A} $}\;
\caption{$ \!{LazyForward} \left( G = \left(V,E \right), R, c, B, type \right) $}
\label{alg:LazyForward}
}
\end{algorithm} 

\subsection{Results}
All experiments are implemented in C++ and conducted on an Ubuntu $18.04.4$ machine with Intel(R) Xeon(R) Gold $6240$C CPU $@2.60$GHz and $251$ GB RAM. 
We conducted the three sets of experiments on four real-world networks whose statistics are presented in Table~\ref{tab:nets}. 

\begin{table}[!htb]
\centering
\caption{The basic structural characteristics of the $4$ networks. Here, $N$ and $M$ are the numbers of nodes and edges, respectively. 
\label{tab:nets}
}
\begin{tabular}{lrrrr}
\hline\hline
Network & $ N \quad $ & $ M \quad $ & Average Degree & Maximum Degree \\
\hline
CA-GrQc	& 5,242	&  14,490 & 5.53	& 81 \\
DBLP & 317081 &  1,049,866&  6.62	& 306 \\
Nethept	& 15,233 & 31,387 & 4.12& 64  \\
com-YouTube	& 1,134,890	& 2,987,624& 5.26 & 28576\\
\hline \hline
\end{tabular}
\end{table}

\subsubsection{Comparing the spreads of $\IC$, $\SIR$, and $\TSIR$}
\label{sec:empiricalIC-SIR}
In this section, we validate that in the graph $G$, given the same seeds $S$ ($S$ is obtained by $\IMM$ algorithm with $\epsilon=0.5$ and $\ell=1$ ), the influence spread of $S$ in $\IC$ dominates its counterpart in $\SIR$ under corresponding parameters satisfied \Cref{eqn:IC-SIR}. 
Here, we set $p\in\{0.01,0.05,0.1,0.3,0.5\}$ for $\IC$, the recovery rate is set to $0.8$ in $\SIR$ and $\TSIR$, and the infected rate is calculated according to \Cref{eqn:IC-SIR}. Specifically, the influence spread of $S$ is measured by the number of nodes that end up in the \emph{active} state in $\IC$, and the number of nodes that end up in the \emph{recovered} states in $\SIR$ after the propagation process. 
For $\TSIR$, $T$ is set to $T=100$.
For a given seed set, the expected spread is measured by 10000 times of Monte-Carlo simulations. The only exception is that of com-YouTube where we set the number of Monte Carlo as $1000$, due to the graph's large scale.

\begin{figure}[!h] 
    \centering  
    \begin{minipage}{0.49\linewidth}
        \centering
        \subfigure[CA-GrQc]{
            \label{CA-GrQcICandSIR}
            \includegraphics[width=\textwidth]{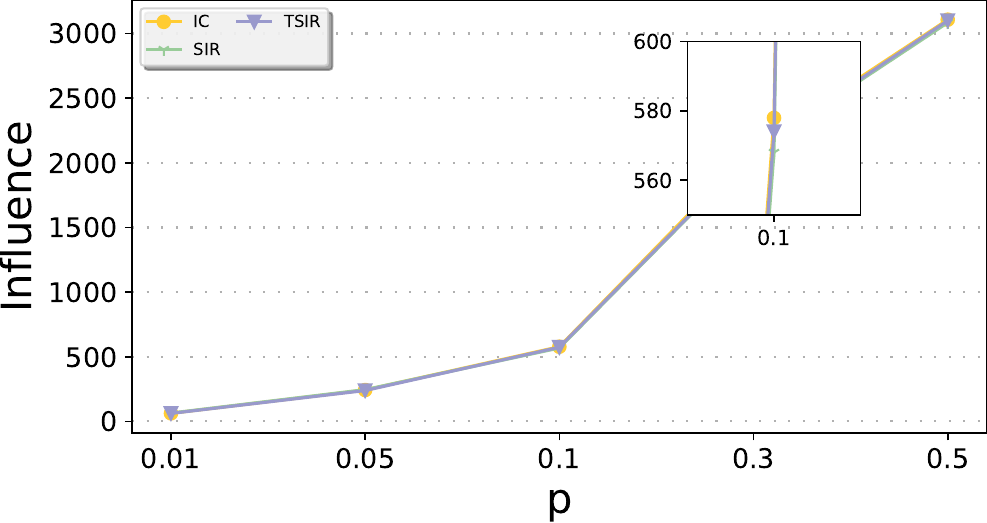}}
    \end{minipage}
    \begin{minipage}{0.49\linewidth}
        \centering
        \subfigure[DBLP]{
            \label{DBLPICandSIR}
            \includegraphics[width=\textwidth]{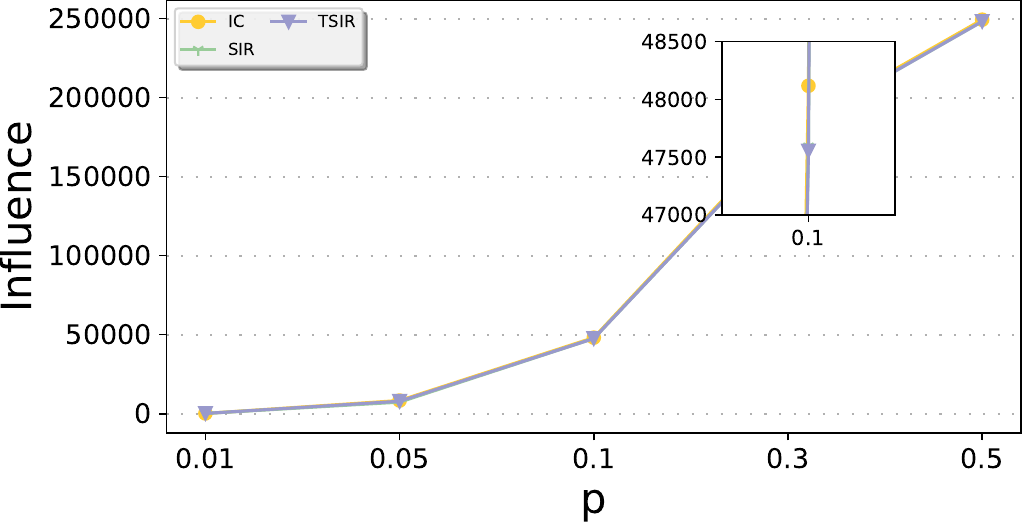}}
    \end{minipage}
    \begin{minipage}{0.49\linewidth}
        \centering
        \subfigure[Nethept]{
            \label{NetheptICAndSIR}
            \includegraphics[width=\textwidth]{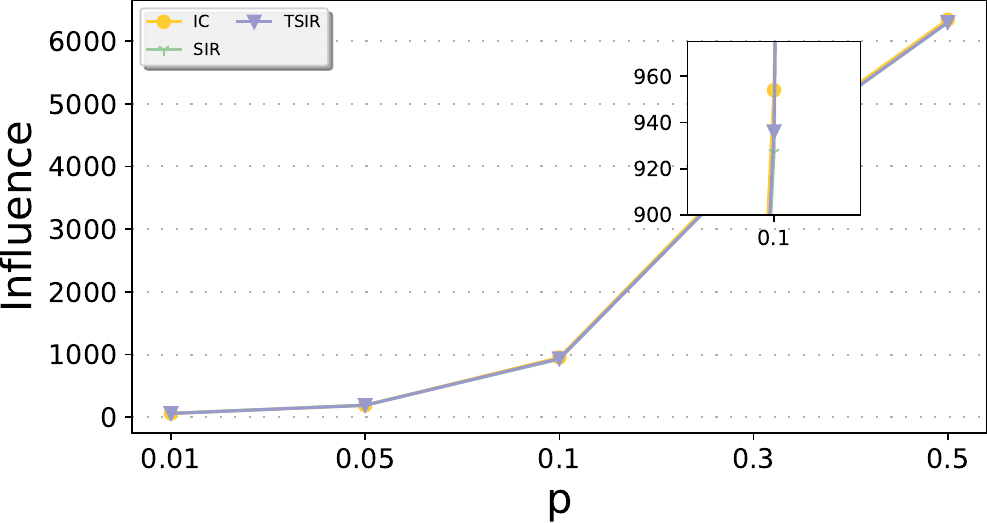}}
    \end{minipage}
    \begin{minipage}{0.49\linewidth}
        \centering
        \subfigure[com-YouTube]{
            \label{com-youtubeICandSIR}
            \includegraphics[width=\textwidth]{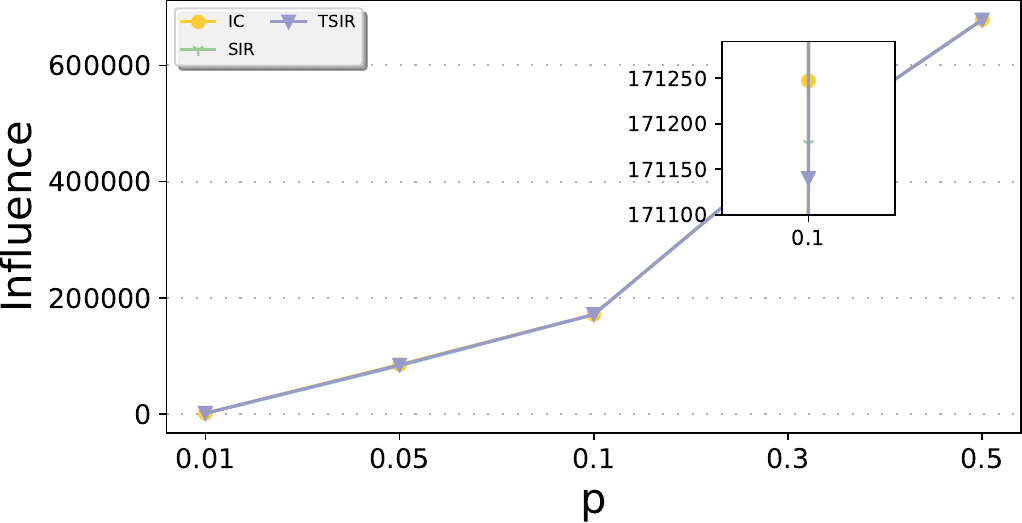}}
    \end{minipage}
    \caption{The performances of the same $S$ under $\IC$ and $\SIR$, with a set of matching parameters $p$, $\beta$ and $\gamma$, respectively. Due to the significant overall fluctuations in the data, we use an inset plot to highlight the changes within a specific range. Without loss of generality, we show the specific fluctuation of Influence on the point when $p=0.1$.}
	\label{fig: ICAndSIR}
\end{figure}

Figure~\ref{fig: ICAndSIR} demonstrates the influence spread of $S$ under $\IC$, $\SIR$, and $\TSIR$ across the different values of $p$. As expected, the result shows that the influence of $S$ on $\IC$ dominates $\SIR$ and $\TSIR$ in real-world datasets, but by not a very large margin in general. As analyzed in \Cref{section:IC-and-SIR}, the positive correlation in $\SIR$ is the key factor causing this dominance. It is evident that in networks with dense edges, such as DBLP, this dominance is pronounced due to the significant impact of positive correlation in the $\SIR$ model.

\subsubsection{Comparison of Algorithms by Seeds Qualities}
In this section, we compare our proposed algorithms $\SIRIMM$ and $\TSIRIMM$ with the four baseline algorithms $\CELF$, $\DegDis$, $\IMM$ (by treating the $\SIR$/$\TSIR$ diffusion model as the $\IC$ model according to \Cref{eqn:IC-SIR}), and degree centrality.
For the diffusion model parameters, we set $p = 0.1$, $\gamma = 0.8$, and $\beta$ is calculated according to \Cref{eqn:IC-SIR}.
The parameters we used for these algorithms are shown in \Cref{tab:parameters}.

\begin{table}[h]
\centering
\caption{Parameters for baseline approach. \label{tab:parameters}}
\small
\begin{tabular}{lcc}
\hline\hline
Algorithms & Parameter & Value \\
\hline
\hline
CELF & MC Simulations & 10000 \\
\hline
Degree Discount & $p$	& 0.1 \\
\hline
$\IMM$ & 
    \begin{tabular}{c}
         $\epsilon$  \\
         $\ell$ 
    \end{tabular}
& \begin{tabular}{c}
         $0.5$  \\
         $1$ 
    \end{tabular} \\
\hline
$\SIRIMM$	& \begin{tabular}{c}
         $\epsilon$  \\
         $\ell$ 
    \end{tabular}
& \begin{tabular}{c}
         $0.5$  \\
         $1$ 
    \end{tabular} \\
\hline
$\TSIRIMM$ & \begin{tabular}{c}
         $\epsilon$  \\
         $\ell$ 
    \end{tabular}
& \begin{tabular}{c}
         $0.5$  \\
         $1$ 
    \end{tabular} \\

\hline \hline
\end{tabular}
\end{table}

We proceed to evaluate their performance by benchmarking their qualities. In our specific problem, quality refers to the expected spread generated by the $k$ influential seeds identified by each algorithm, and the expected spread is measured by 10000 times of Monte-Carlo simulations. The only exception is that of com-YouTube where we set the number of Monte Carlo as $1000$, due to the graph's large scale.

\paragraph{Compare $\SIRIMM$ with baselines.} 
Here we focus on the performance of seed in the $\SIR$ model by algorithms includes $\SIRIMM$, $\DegDis$, $\IMM$, degree centrality, and $\CELF$.

\begin{figure}[!h] 
    \centering  
    \begin{minipage}{0.49\linewidth}
        \centering
        \subfigure[CA-GrQc]{
            \includegraphics[width=\textwidth]{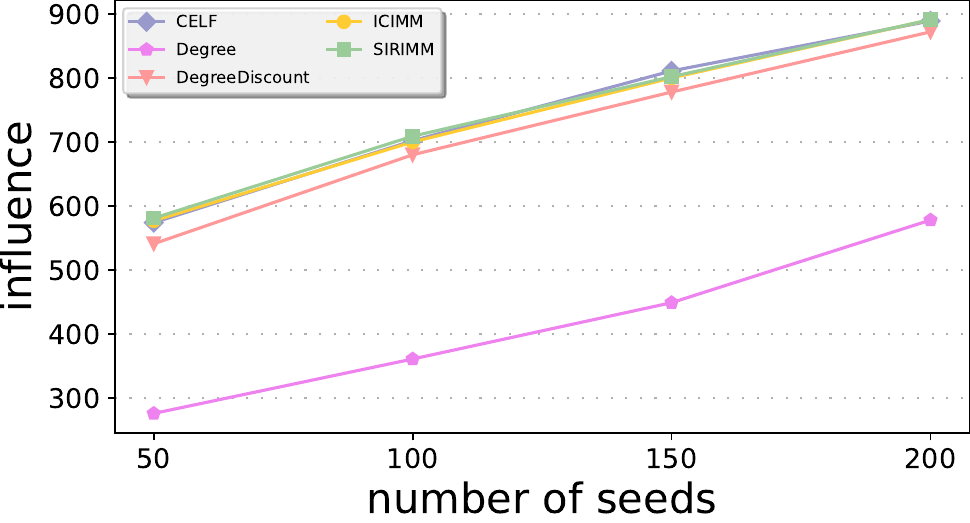}}
    \end{minipage}
    \begin{minipage}{0.49\linewidth}
        \centering
        \subfigure[DBLP]{
            \includegraphics[width=\textwidth]{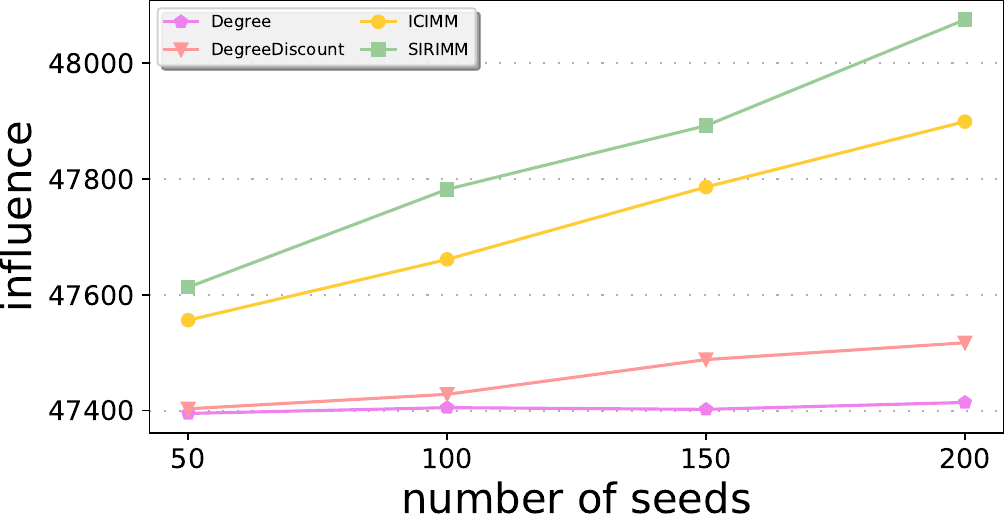}}
    \end{minipage}
    \begin{minipage}{0.49\linewidth}
        \centering
        \subfigure[Nethept]{
            \includegraphics[width=\textwidth]{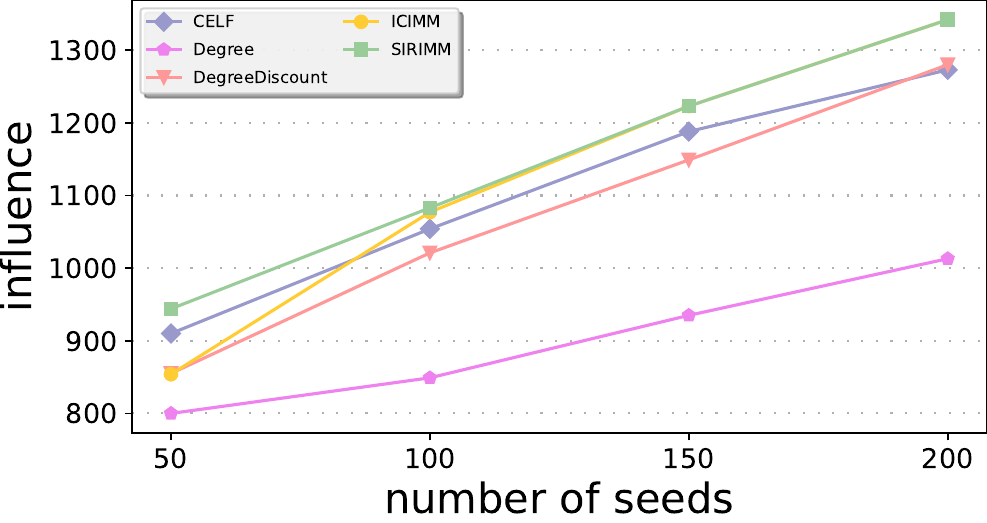}}
    \end{minipage}
    \begin{minipage}{0.49\linewidth}
        \centering
        \subfigure[com-YouTube]{
            \includegraphics[width=\textwidth]{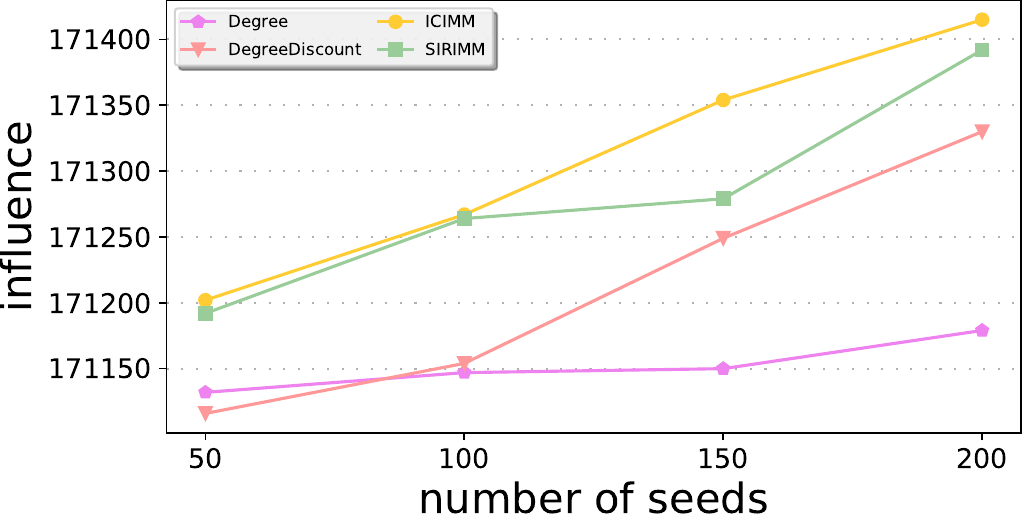}}
    \end{minipage}
    \caption{Comparing the influence spread of $\SIRIMM$ with baseline methods.}
    \label{fig:influence}
\end{figure}

As shown in~\Cref{fig:influence}, the influence spread of $\SIRIMM$ is higher than that of the baseline method in most cases. 
Although $\CELF$ and $\IMM$ are designed to greedily select the influential seeds, both of them are assumed based on the $\IC$ model, specifically, both of them try to approximate the result of $\SIR$ based on the $\IC$ model. Conversely, $\SIRIMM$ is designated tailored to the $\SIR$ model; naturally, the performance of $\SIRIMM$ is higher than its counterpart of $\CELF$ and $\IMM$. Notice that, this difference is slight,
the reason is akin to the analysis in \Cref{sec:empiricalIC-SIR}. As for $\DegDis$ and degree centrality, both of them are proposed to select seeds in a general way and have strong scalability, but failed to target on specific diffusion models.


\paragraph{Compare $\TSIRIMM$ with baselines.}  Here we focus on the performance of seed selection in the $\TSIR$ model by algorithms including $\DegDis$, $\IMM$, degree centrality, and $\CELF$.

\begin{figure}[!h] 
    \centering  
    \begin{minipage}{0.49\linewidth}
        \centering
        \subfigure[CA-GrQc]{
            \includegraphics[width=\textwidth]{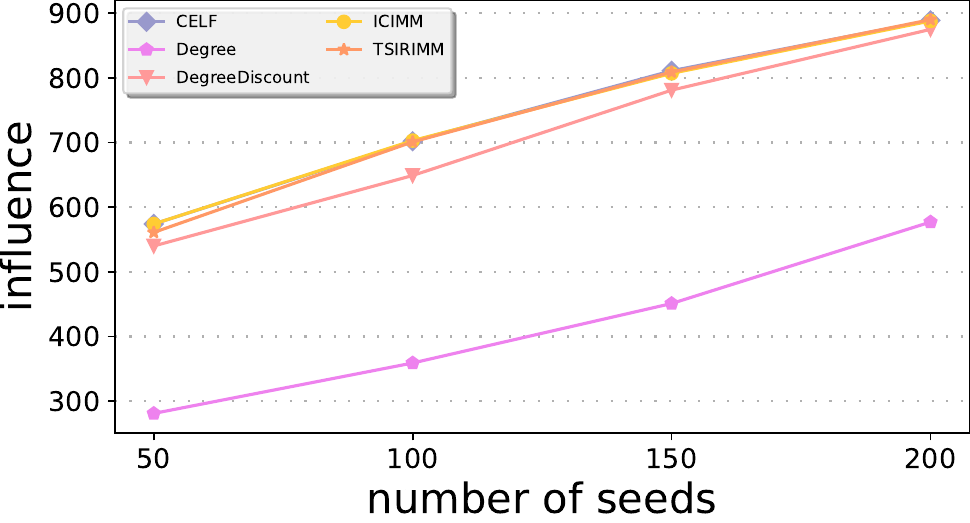}}
    \end{minipage}
    \begin{minipage}{0.49\linewidth}
        \centering
        \subfigure[DBLP]{
            \includegraphics[width=\textwidth]{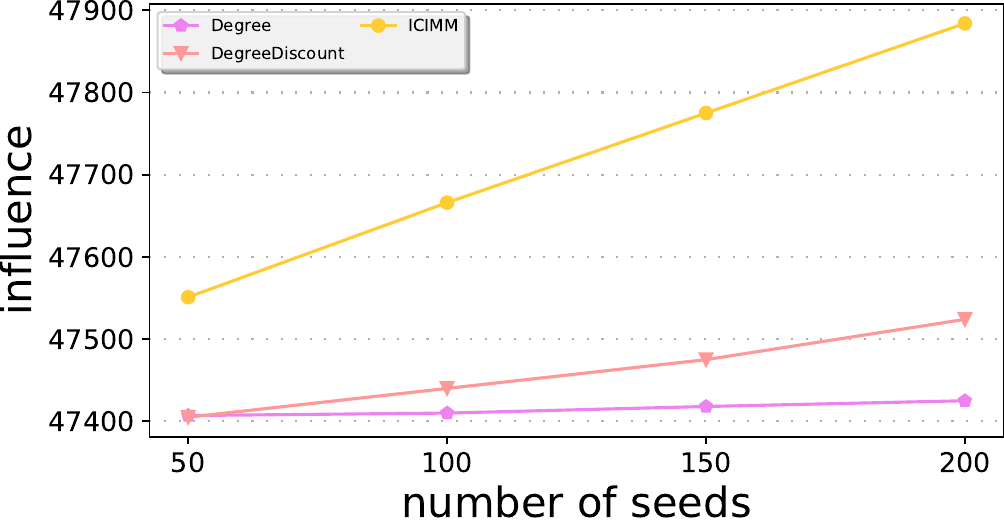}}
    \end{minipage}
    \begin{minipage}{0.49\linewidth}
        \centering
        \subfigure[Nethept]{
            \includegraphics[width=\textwidth]{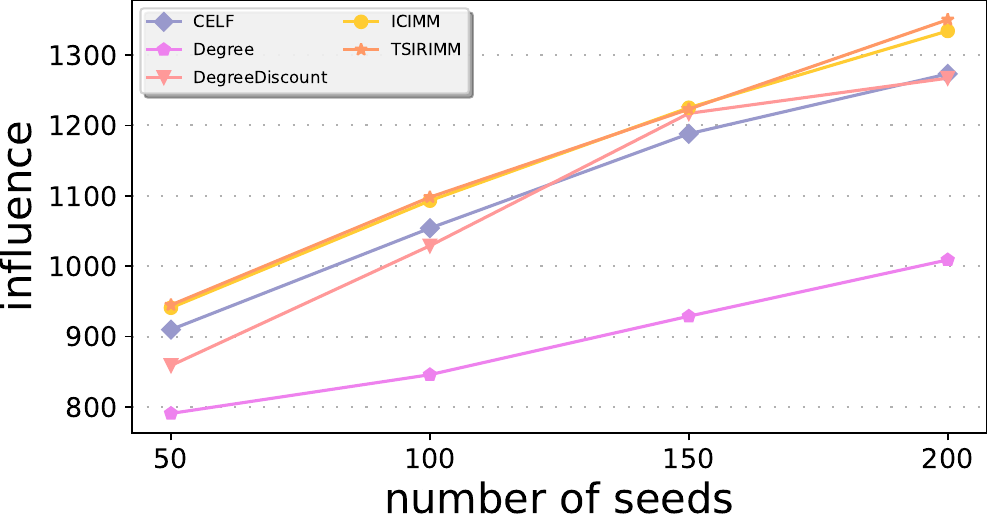}}
    \end{minipage}
    \begin{minipage}{0.49\linewidth}
        \centering
        \subfigure[com-YouTube]{
            \includegraphics[width=\textwidth]{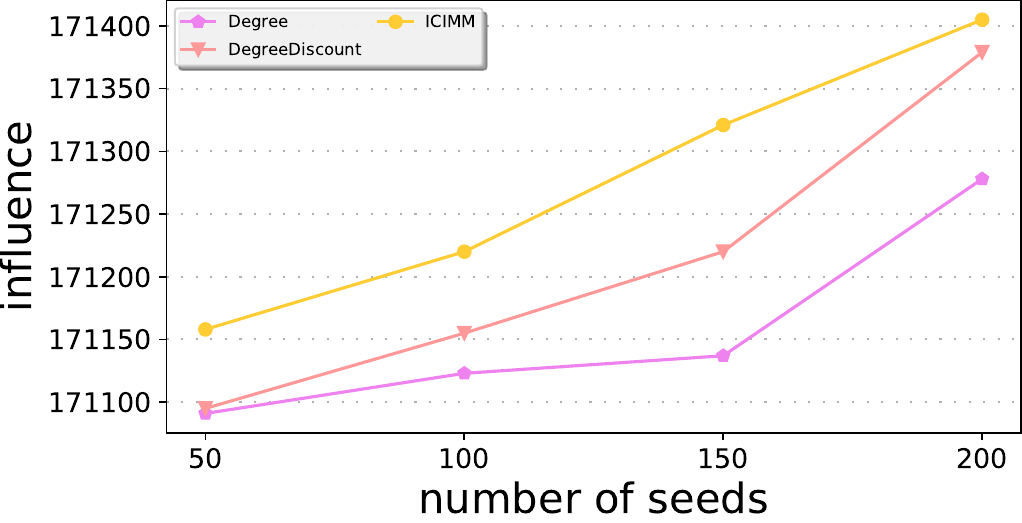}}
    \end{minipage}
    \caption{Comparing the influence spread of $\TSIRIMM$ with baseline methods.}
    \label{fig:influenceOfTSIR}
\end{figure}

As shown in~\Cref{fig:influenceOfTSIR}, the influence spread of $\TSIRIMM$ is higher than that of the baseline method in CA-GrQc and Nethept based on the $\TSIR$ model, which is because $\TSIRIMM$ is specifically designed based on $\TSIR$ model while the others are not.
However, due to the extra complications of computing shortest paths in the $\RR$ set sampling process (caused by the time-dependency feature of the $\TSIR$ model), $\TSIRIMM$ has a much lower scalability. In particular, it fails to scale to DBLP and com-YouTube (the algorithm did not terminate after more than 48 hours). The algorithm $\CELF$ also fails to scale to DBLP and com-Youtube as the Monte-Carlo sampling used by $\CELF$ is generally slower than $\RR$ sets sampling.


\subsubsection{Comparison of Algorithms by Running Times}

In this analysis, we investigate the running time of our algorithms ($\SIRIMM$ and $\TSIRIMM$), greedy method $\CELF$, $\DegDis$, $\IMM$, and degree centrality.
The parameter $T$ for the $\TSIR$ model (and thus the algorithm $\TSIRIMM$) is set to $T=50$.

\begin{figure}[!h] 
    \centering  
    \vspace{-0.35cm} 
    \subfigtopskip=2pt 
    \subfigbottomskip=2pt 
    \subfigcapskip=-5pt 
    \begin{minipage}{0.49\linewidth}
        \centering
        \subfigure[CA-GrQc]{
            \includegraphics[width=\textwidth]{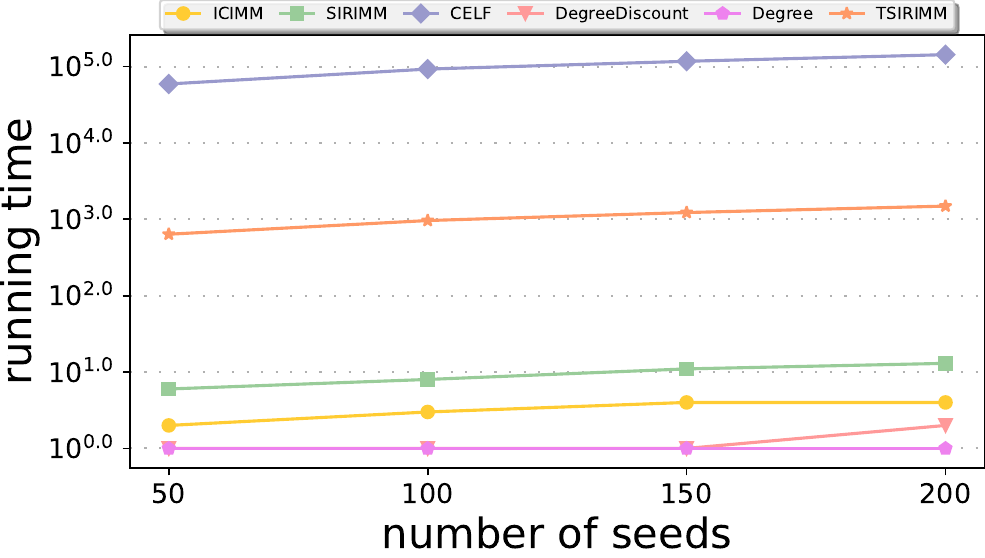}}
    \end{minipage}
    \begin{minipage}{0.49\linewidth}
        \centering
        \subfigure[DBLP]{
            \includegraphics[width=\textwidth]{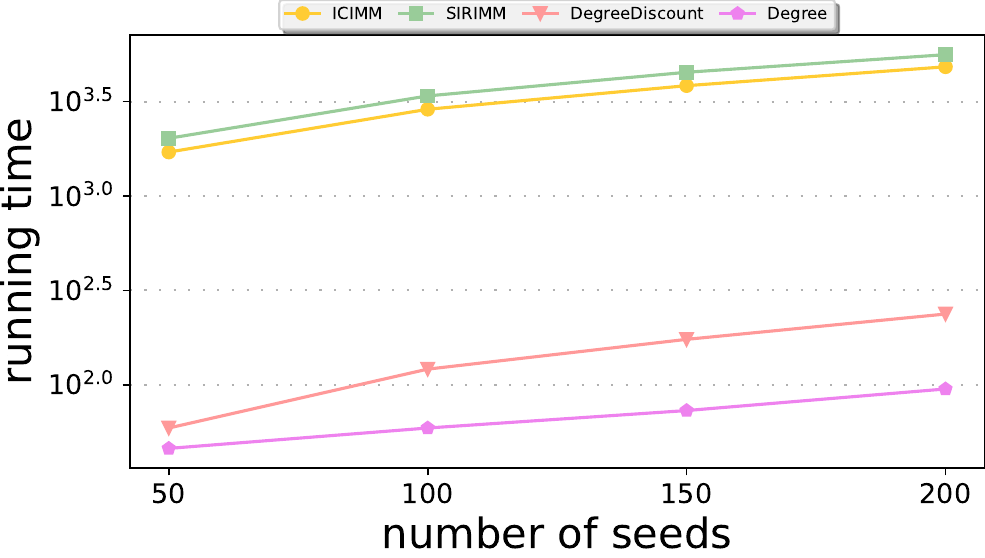}}
    \end{minipage}
    \begin{minipage}{0.49\linewidth}
        \centering
        \subfigure[Nethept]{
            \includegraphics[width=\textwidth]{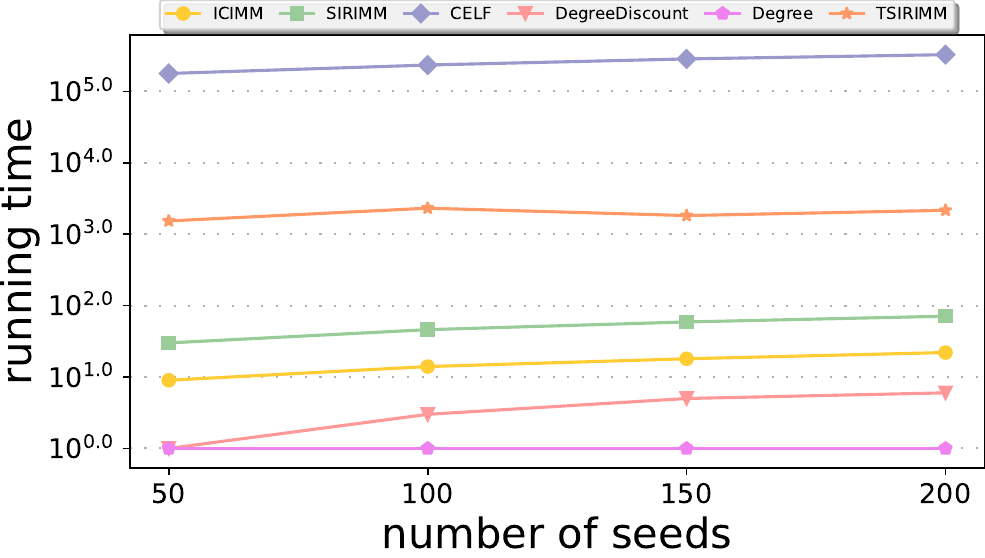}}
    \end{minipage}
    \begin{minipage}{0.49\linewidth}
        \centering
        \subfigure[com-YouTube]{
            \includegraphics[width=\textwidth]{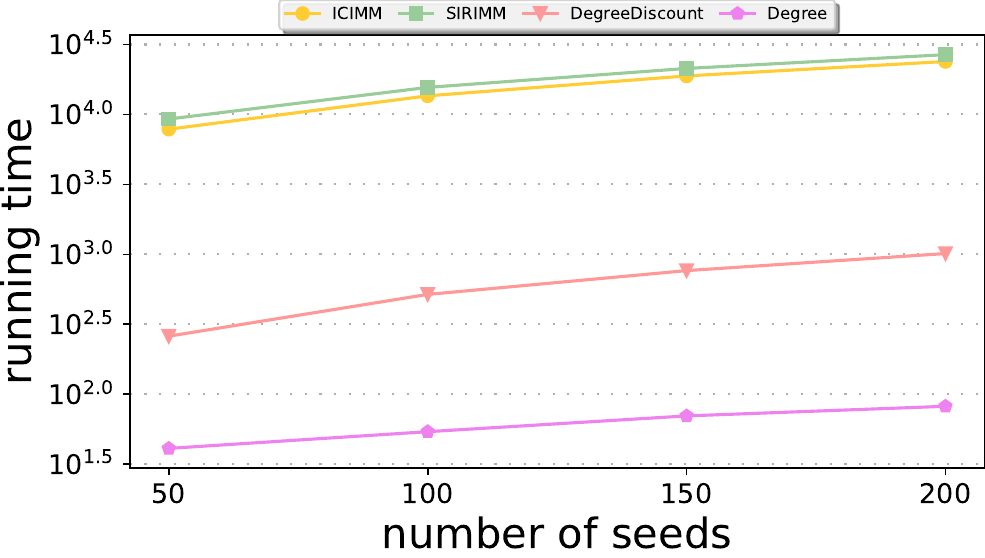}}
    \end{minipage}
    \caption{Comparing the scalablity of $\SIRIMM$, $\TSIRIMM$ with greedy and heuristics.}
    \label{fig:running-time}
\end{figure}
As shown in~\Cref{fig:running-time}, the running time of our proposed $\TSIRIMM$ and $\SIRIMM$ and the baseline methods are compared. As mentioned before, the running times of $\TSIRIMM$ and $\CELF$ on DBLP and com-YouTube are missing. 

In CA-GrQc and Nethept, the running time of $\CELF$ is significantly higher than that of the other methods ($\IMM$, $\SIRIMM$, $\DegDis$, and degree centrality).  
As expected, $\TSIRIMM$ also has a higher running time than the other methods. However, the gap between $\TSIRIMM$ and the other methods is smaller than the gap between $\CELF$ and the other methods. In DBLP and com-YouTube, the running time of $\SIRIMM$ is slightly higher than the remaining baselines, which is because of the $\RR$ set sampling process in the $\SIRIMM$ method.

\section{Conclusion}
In this paper, we have explored the similarities and differences between the two diffusion models: $\IC$ and $\SIR$.
Prior work has observed that the $\SIR$ model can be ``approximately'' viewed by the $\IC$ model when the $\IC$ parameter of each edge $(u,v)$ is viewed as an aggregated probability that $u$ eventually infects $v$ in the $\SIR$ model (see \Cref{eqn:IC-SIR}).
However, if viewed in this way, there is a dependency on the outgoing edges of every vertex in the $\SIR$ model.
In this paper, we have performed an in-deep analysis on this dependency both theoretically and empirically. 
We have proved that this extra dependency always harms the spread of influence.
As a result, under matching parameters, given the same seed set, the influence spread under the $\IC$ model dominates that under the $\SIR$ model.
We have also shown that this domination can be significant in the theoretical worst-case analysis, while it is less significant in our empirical experiments.

By exploiting some of our proof techniques, we adapt the $\IMM$ algorithm to the $\InfMax$ problem under the $\SIR$ diffusion model and the $\TSIR$ model (which is a time-dependent variant of the $\SIR$ model).
To the best of our knowledge, this is the first $\InfMax$ algorithm for the $\SIR$ model with the $(1-1/e)$ theoretical guarantee.

\section*{Acknowledgments}
The research of Biaoshuai Tao was supported by the National Natural Science Foundation of China (No. 62102252). The research of Kuan Yang was supported by the NSFC grant No. 62102253.

\bibliographystyle{plainnat}  
\bibliography{ref}  

\newpage
\appendix
\section{Omitted Proofs}
\label{sec:appendix}
\liveedgeSIR*
\begin{proof}
    The proof is just the rephrase of the stochastic diffusion process. We drive the diffusion process by the collection of random variables $\set{\boldsymbol{R}_v}_{v\in V},\set{\boldsymbol{I}_e}_{e\in E}$ that generate the live-edge graph as mentioned in \Cref{subsec:live-edge-SIR}. 
    
    The diffusion process can be driven by $\set{\boldsymbol{R}_v}_{v\in V},\set{\boldsymbol{I}_e}_{e\in E}$ as follows.
    Initially, each node $v\in V$ maintains an index $j_v=1$, and the nodes in the seed set $S$ are labeled as infected at timestamp $0$. At each timestamp $t=1,2,\dots$, each infected node $u$ from the previous timestamp $t-1$ performs the following operations sequentially:
     \begin{itemize}
         \item for each of its susceptible outgoing neighbors $v$, $u$ infects $v$ if $I_{(u,v),j_u}=1$;
         \item $u$ gets recovered if $R_{u,j_u}=1$ and remains infected otherwise;
         \item $j_u\gets j_u+1$.
     \end{itemize}
     Examining the randomness for any infected node $u$ in the above process, one can verify that there exists a path in $\+G_{\SIR}$ from $S$ to $u$.
     Conversely, regarding the randomness for any path from $S$ to the node $u$ in $\+G_{\SIR}$, we can recover the diffusion process in which $u$ gets infected.
\end{proof}

\NegativeCorrelatde*
\begin{proof}
   Suppose that $e$ is an outgoing edge of $u$ in the underlying graph $G$, and let $E'_u$ be the collection of outgoing edges of a node $u$ in the set $E'$. If $E'_u=\emptyset$, then by definition we have 
   $$\Pr{e\in \+G_{\SIR} \mid E'\cap \+G_{\SIR}=\emptyset} = \Pr{e\in \+G_{\SIR}}.$$ Therefore, it suffices to consider the case when $E'_u\neq \emptyset$. Note that
    \begin{align*}
        \Pr{e\in \+G_{\SIR} \mid E'\cap \+G_{\SIR}=\emptyset}=\Pr{e\in \+G_{\SIR} \mid E'_u \cap \+G_{\SIR}=\emptyset)}=\frac{\Pr{e\in \+G_{\SIR} , E'_u \cap \+G_{\SIR}=\emptyset)}}{\Pr{E'_u \cap \+G_{\SIR}=\emptyset}}.
    \end{align*}  In the following, we prove 
    \begin{align}
    \label{eqn:correlation-inequality}
    \Pr{e\in \+G_{\SIR} , E'_u \cap \+G_{\SIR}=\emptyset)}\leq  \Pr{e\in \+G_{\SIR}}\cdot \Pr{E'_u\cap \+G_{\SIR}=\emptyset}
    \end{align}by the coupling argument, which implies \Cref{lem:negative-correlated} immediately. Specifically, we introduce a unified probability space and demonstrate that the aforementioned terms are equivalent to the probability of specific events that can be defined in this space.

   Let $\set{\boldsymbol{R}_v}_{v\in V}$, $\set{\boldsymbol{R}'_v}_{v\in V}$ and $ \set{\boldsymbol{I}_f}_{f\in E}$ be collections of independent random variables as mentioned in~\Cref{sect:liveandRRset}. We define the graphs $\+G_1$ and $\+G_2$ as: 
   \[
        \+G_1:= \+G_{\SIR}(\set{\boldsymbol{R}_v}_{v\in V},\set{\boldsymbol{I}_f}_{f\in E}), \ \
        \+G_2:= \+G_{\SIR}(\set{\boldsymbol{R}'_v}_{v\in V},\set{\boldsymbol{I}_f}_{f\in E}).
    \] We can then observe that \Cref{eqn:correlation-inequality} is equivalent to 
    \begin{align*}
         \Pr{e\in \+G_1 , E'_u\cap \+G_1=\emptyset }\leq \Pr{ e\in \+G_1,E'_u \cap \+G_2=\emptyset},
    \end{align*}
    since the events $\set{e\in \+G_1}$ and $\set{E'_u\cap \+G_2=\emptyset}$ are independent.
    
    For each edge $f\in E$, let $T_{f}$ be the first time that $I_{f,T_{f}}=1$, i.e., $I_{f,T_{f}}=1$ and $I_{f,[T_{f}-1]}$ is a sequence of $0$ with length $T_{f}-1$. We also call $T_{f}$ as the infection time of edge $f\in E$. The event $\set{e\in \+G_1 , E'_u\cap \+G_1=\emptyset}$ happens only if the infection time of $e$ is earlier than the infection time of any other edges in $E'_u$, i.e., $T_e<T_{f} $ for any $f\in E'_u$. For simplicity, let
    $$\+E_1:= \set{e\in \+G_1 , E'_u\cap \+G_1=\emptyset ,T_e< T_{f}, \forall f\in E'_u},$$
    and
    $$\+E_2:= \set{e\in \+G_1 , E'_u\cap \+G_2=\emptyset ,T_e< T_{f}, \forall f\in E'_u}.$$  Observe that
    \begin{align*}
         \Pr{e\in \+G_1 , E'_u\cap \+G_1=\emptyset}= \Pr{\+E_1},\ \
        \Pr{\+E_2}\leq \Pr{ e\in \+G_1 , E'_u\cap \+G_2=\emptyset}.
    \end{align*}
    Let $T^*:=\min_{f\in E'_u} T_{f}$ and 
    \begin{align*}
        &\!{P}_{1}(t_1,t_2):= \Pr{e\in \+G_1 , E'_u\cap \+G_1=\emptyset ,T_e=t_1,T^*=t_2},\\
        &\!{P}_{2}(t_1,t_2):= \Pr{ e\in \+G_1 , E'_u\cap \+G_2=\emptyset,T_e=t_1,T^*=t_2},
    \end{align*} for any positive integer $1\leq t_1< t_2$.
    It boils down to show $\!{P}_{1}(t_1,t_2)\leq   \!{P}_{2}(t_1,t_2)$ since 
    \begin{align*}
        \Pr{\+E_1}=\sum_{t_1=1}^{\infty} \sum_{t_2=t_1+1}^{\infty} \!{P}_{1}(t_1,t_2)\quad\mbox{ and }\quad
        \Pr{ \+E_2}=\sum_{t_1=1}^{\infty} \sum_{t_2=t_1+1}^{\infty} \!{P}_{2}(t_1,t_2).
    \end{align*}
    By definition, we have
    \begin{align*}
          &\quad \!{P}_{1}(t_1,t_2)=\sum_{t=t_1}^{t_2-1}(1-\gamma_v)^{t-1}\cdot \gamma_v\cdot \Pr{T_e=t_1,T^*=t_2}.
    \end{align*} Meanwhile, we have
    \begin{align*}
         \!{P}_2(t_1,t_2)&=\Pr{ e\in \+G_1 , E'_u\cap \+G_2=\emptyset,T_e=t_1,T^*=t_2}\\
         &=\sum_{t=t_1}^{\infty}\sum_{t'=1}^{t_2-1}(1-\gamma_v)^{t+t'-2}\cdot \gamma_v^2\cdot \Pr{T_e=t_1,T^*=t_2}\notag \\
         &=\sum_{t'=1}^{t_2-1}(1-\gamma_v)^{t_1+t'-2}\cdot \gamma_v\cdot \Pr{T_e=t_1,T^*=t_2}\notag\\
         &\geq  \sum_{t=t_1}^{t_2-1}(1-\gamma_v)^{t-1}\cdot \gamma_v\cdot \Pr{T_e=t_1,T^*=t_2}= \!{P}_{1}(t_1,t_2).
    \end{align*}

    Combining all these facts, the proof is complete.
\end{proof}

\ICDominateSIRGreatly*
\begin{proof}
    Since the parameters satisfying \Cref{eqn:IC-SIR}, we can calculate the $\IC$ parameter $p$ for the dashed edges as follows 
    \begin{align}\label{eqn:equalinfluencerate}
        p=\sum_{t=1}^{\infty}(1-\gamma )^{(t-1)}\cdot \gamma \cdot (1-(1-\beta)^t)=1-\frac{\gamma(1-\beta)}{\gamma+\beta-\gamma\cdot\beta},
    \end{align}
    where the last equality follows from standard computations of geometric series.
    
    Notice that, by our definition of solid edges, if at least one of the in-neighbors of $u$ is infected, $u$ will be infected with probability $1$, and this is true under both models.
    Let $p_1(b)$ be the probability that the node $u$ gets infected in $\IC$ and $p_2(b)$ be that in $\SIR$. One can easily verify that 
    \[
        p_1(b)=1-(1-p)^{b}=
        \notag
        \\1-\left(\frac{\gamma\cdot (1-\beta)}{\gamma+\beta-\gamma\cdot \beta}\right)^b
    \] according to \Cref{eqn:equalinfluencerate}, and
    \begin{align*}
        p_2(b)&=1-\sum_{t=1}^{\infty}(1-\gamma )^{(t-1)}\cdot \gamma\cdot  (1-(1-\beta)^{b\cdot t})=\frac{\gamma\cdot (1-\beta)^b}{1-(1-\gamma)\cdot (1-\beta)^b}.
    \end{align*}
    Setting $\gamma=b^{-0.5}$ and $\beta=b^{-1.5}$, 
    \begin{align*}
        p_1(b)&=1-\left(\frac{b^{-0.5}\cdot (1-b^{-1.5})}{b^{-0.5}+b^{-1.5}-b^{-2}}\right)^b\\
        &\leq  1-\Bigl(\frac{b^{-0.5}}{b^{-0.5}+2b^{-1.5}}\Bigr)^b\\
        &=1-\Bigl(\frac{2b}{1+2b}\Bigr)^b
        =1-\Bigl(1-\frac{1}{1+2b}\Bigr)^b,
    \end{align*} where the second inequality holds when $b^{-1.5}\leq \frac{1}{2}$. Therefore, $p_1(b)$ tends to $1-e^{-1/2}$ when $b$ goes to infinity.
    On the other hand, 
    \begin{align*}
        p_2(b)&=\frac{b^{-0.5}\cdot (1-b^{-1.5})^b}{1-(1-b^{-0.5})\cdot (1-b^{-1.5})^b} =
        \frac{b^{-0.5}\cdot (1-b^{-1.5})^b}{1-(1-b^{-1.5})^b +b^{-0.5}\cdot (1-b^{-1.5})^b}.
    \end{align*} 
    Hence, $p_2(b)$ tends to $0$ when $b$ goes to infinity since $b^{-0.5}\cdot (1-b^{-1.5})^b$ tends to $0$.
    
    Finally, under both models, those $n_0$ vertices on the right-hand side in \Cref{fig:gadget} will be infected with probability $1$ if $u$ is infected. Therefore, we have
    \begin{equation}\label{eqn:ICtest}
        \begin{aligned}
            \sigma_{{\IC}}(v)&=\sum_{u'\in V} \Pr{\mbox{$u'$ gets influenced by $v$}}\geq  n_0\cdot p_1(b), 
        \end{aligned}
    \end{equation}
    and 
    \begin{equation}\label{eqn:SIRtest}
        \begin{aligned}
            \sigma_{{\SIR}}(v)
            &=\sum_{u'\in V} \Pr{\mbox{$u'$ gets influenced by $v$}}
            \leq (b+1)+ n_0\cdot p_2(b).
        \end{aligned}
    \end{equation}
    Combining \Cref{eqn:ICtest} and \Cref{eqn:SIRtest},  
    \begin{align*}
         \frac{\sigma_{{\SIR}}(v)}{\sigma_{{\IC}}(v)}\leq \frac{(b+1)+ n_0\cdot p_2(b)}{n_0\cdot p_1(b)}.
    \end{align*}
    Since $p_1(b)$ tends to a positive constant $1-e^{-1/2}$ and $p_2(b)$ tends to $0$ as $b\rightarrow\infty$, the ratio $\frac{\sigma_{{\SIR}}(v)}{\sigma_{{\IC}}(v)}$
    tends to $0$ when $b\ll n_0$ and $b$ goes to infinity (for example, we can set $n_0=b^2$ and let $b\rightarrow\infty$).
\end{proof}

\end{document}